\documentclass[11pt]{article}
\usepackage[centertags]{amsmath}
\usepackage{amsfonts}
\usepackage{amssymb}
\usepackage{amsthm}
\usepackage{amsmath}
\usepackage{version}
\usepackage{tabularx}
\usepackage{bbm}
\usepackage{algorithm}
\usepackage{algorithmic}
\usepackage{placeins}
\usepackage{subfigure}
\usepackage{slashbox}
\usepackage[german,english]{babel}
\usepackage{latexsym,graphicx}
\usepackage{color}

\theoremstyle{plain}
\newtheorem{thm}{Theorem}[section]

\newtheorem{prop}[thm]{Proposition}

\theoremstyle{definition}

\theoremstyle{plain}
\newtheorem{example}{\bf Example}

\usepackage{fancyhdr}

\makeatletter
\def \newequation#1#2{
   \@definecounter{#1}
   \@namedef{the#1}{\hbox{#2}}
   \@namedef{#1}{$$\refstepcounter{#1}}
   \@namedef{end#1}{
      \eqno \csname the#1\endcsname $$\global\@ignoretrue
      }
   }
\makeatother
\newequation{hyp}{(A\arabic{hyp})}

\makeatletter
\def \newequation#1#2{
   \@definecounter{#1}
   \@namedef{the#1}{\hbox{#2}}
   \@namedef{#1}{$$\refstepcounter{#1}}
   \@namedef{end#1}{
      \eqno \csname the#1\endcsname $$\global\@ignoretrue
      }
   }
\makeatother
\newequation{hyp1}{$(A2)_{a}$}

\makeatletter
\def \newequation#1#2{
   \@definecounter{#1}
   \@namedef{the#1}{\hbox{#2}}
   \@namedef{#1}{$$\refstepcounter{#1}}
   \@namedef{end#1}{
      \eqno \csname the#1\endcsname $$\global\@ignoretrue
      }
   }
\makeatother
\newequation{hyp4}{$(A3)$}

\makeatletter
\def \newequation#1#2{
   \@definecounter{#1}
   \@namedef{the#1}{\hbox{#2}}
   \@namedef{#1}{$$\refstepcounter{#1}}
   \@namedef{end#1}{
      \eqno \csname the#1\endcsname $$\global\@ignoretrue
      }
   }
\makeatother
\newequation{hyp2}{(B\arabic{hyp2})}

\makeatletter
\def \newequation#1#2{
   \@definecounter{#1}
   \@namedef{the#1}{\hbox{#2}}
   \@namedef{#1}{$$\refstepcounter{#1}}
   \@namedef{end#1}{
      \eqno \csname the#1\endcsname $$\global\@ignoretrue
      }
   }
\makeatother
\newequation{hyp3}{($\mathcal{H}_{b,\sigma}$)}

\makeatletter
\def \newequation#1#2{
   \@definecounter{#1}
   \@namedef{the#1}{\hbox{#2}}
   \@namedef{#1}{$$\refstepcounter{#1}}
   \@namedef{end#1}{
      \eqno \csname the#1\endcsname $$\global\@ignoretrue
      }
   }
\makeatother
\newequation{E1}{($E_{b,\sigma}$)}

\makeatletter
\def \newequation#1#2{
   \@definecounter{#1}
   \@namedef{the#1}{\hbox{#2}}
   \@namedef{#1}{$$\refstepcounter{#1}}
   \@namedef{end#1}{
      \eqno \csname the#1\endcsname $$\global\@ignoretrue
      }
   }
\makeatother
\newequation{hyppsi}{($H_{\delta}^{es}$)}

\begin{document}

    \vspace{\stretch{1}}

   \title{Computing VaR and CVaR using Stochastic Approximation and Adaptive Unconstrained Importance Sampling}
          
          \author{{\normalsize O. Bardou${}^{1}$  , N. Frikha${}^{2}$ , G. Pag\`{e}s${}^{3}$}} 
    \maketitle
   \begin{abstract}
  Value-at-Risk (VaR) and Conditional-Value-at-Risk (CVaR) are two risk measures which are widely used in the practice of risk management. This paper deals with the problem of estimating both VaR and CVaR using stochastic approximation (with decreasing steps): we propose a first Robbins-Monro (RM) procedure based on Rockafellar-Uryasev's identity for the CVaR. Convergence rate of this algorithm to its target satisfies a Gaussian Central Limit Theorem. As a second step, in order to speed up the initial procedure, we propose a recursive and adaptive importance sampling (IS) procedure which induces a significant variance reduction of both VaR and CVaR procedures. This idea, which has been investigated by many authors, follows a new approach introduced in \cite{lemairepages}. Finally, to speed up the initialization phase of the IS algorithm, we replace the original confidence level of the VaR by a slowly moving risk level. We prove that the weak convergence rate of the resulting procedure is ruled by a Central Limit Theorem with minimal variance and its efficiency is illustrated on several typical energy portfolios.\end{abstract}
   
\medskip
\noindent {This work appeared in \emph{Monte Carlo Methods and Applications} 2009.}

\medskip

\footnotetext[{1}]{Laboratoire de Probabilit\'{e}s et Mod\`{e}les al\'{e}atoires, UMR 7599, Universit\'{e} Pierre et Marie Curie and GDF Suez Research and Innovation Department, France, e-mail: olivier-aj.bardou@gdfsuez.com}
\footnotetext[{2}]{Laboratoire de Probabilit\'{e}s et Mod\`{e}les al\'{e}atoires, UMR 7599, Universit\'{e} Pierre et Marie Curie and GDF Suez Research and Innovation Department, France, e-mail: frikha.noufel@gmail.com}
\footnotetext[{3}]{Laboratoire de Probabilit\'{e}s et Mod\`{e}les al\'{e}atoires, UMR 7599, Universit\'{e} Pierre et Marie Curie, France, e-mail: gilles.pages@upmc.fr}

%
\begin{section}{Introduction}

Following financial institutions, energy companies are developing a risk management framework to face the new price and volatility risks associated to the growth of energy markets. Value-at-Risk (VaR) and Conditional Value-at-Risk (CVaR) are certainly the best known and the most common risk measures used in this context, especially for the evaluation of extreme losses potentially faced by traders. Naturally related to rare events, the estimation of these risk measures is a numerical challenge. The Monte Carlo method, which is often the only available numerical device in such a general framework, must always be associated to efficient reduction variances techniques to encompass its slow convergence rate. In some specific cases, Gaussian approximations can lead to semi-closed form estimators. But, if these approximations can be of some interest when considering the {\it yield} of a portfolio, they turn out to be useless when estimating e.g. the VaR on the EBITDA (Earnings Before Interest, Taxes, Depreciation, and Amortization) of a huge portfolio as it is often the case in the energy sector.

In this article, we introduce an alternative estimation method to estmate both VaR and CVaR, relying on the use of recursive stochastic algorithms. By definition, the VaR at level $\alpha\in(0,1)$ (VaR$_{\alpha}$) of a given portfolio is the lowest amount not exceeded by the loss with probability $\alpha$ (usually $\alpha \geq 95\%$). The Conditional Value-at-Risk at level $\alpha$ (CVaR$_{\alpha}$) is the conditional expectation of the portfolio losses beyond the VaR$_{\alpha}$ level. Compared to VaR, CVaR is known to have better properties. It is a coherent risk measure in the sense of Artzner, Delbaen, Eber and Heath, see \cite{artzner}. 

The most commonly used method to compute VaR is the inversion of the simulated empirical loss distribution function using Monte Carlo or historical simulation tools. The historical simulation method usually assumes that the asset returns in the future are independent and identically distributed, having the same distribution as they had in the past. Over a time interval $[t,T]$, the loss is defined by $L:= V(S_{t},t)-V(S_{t}+\Delta S, T)$, where $S_{t}$ denotes the market price vector observed at time $t$, $\Delta S=S_{T}-S_{t}$ the variation of $S$ over the time interval $[t,T]$ -which can be calculated using historical data- and $V(S_{t},t)$ the portfolio value at time $t$. The distribution of this loss $L$ can be computed with the corresponding VaR at a given probability level by the inversion of the empirical function method. However, when the market price dynamics follow a general diffusion process solution of a stochastic differential equation (SDE), the assumption of asset returns independence is no longer available.

	To circumvent this problem, Monte Carlo simulation tools are generally used. Another widely used method relies on a linear (Normal approximation) or quadratic expansion (Delta-Gamma approximation) and assume a joint normal (or log-normal) distribution for $\Delta S$. The Normal approximation method gives $L$ a normal distribution, thus the computation of the VaR$_{\alpha}$ is straightforward. However, when there is a non-linear dependence between the portfolio value and the prices of the underlying assets (think of a portfolio with options) such approximation is no longer acceptable. The Delta-Gamma approximation tries to capture some non linearity by adding a quadratic term in the loss expansion. Then, it is possible to find the distribution of the resulting approximation in order to obtain an approximation of the VaR. For more details about these methods, we refer to \cite{britten}, \cite{duffie}, \cite{glasserman}, \cite{glassermanb} and \cite{rouvinez}. Such approximations are no longer acceptable when considering portfolios with long maturity ($T-t=$ 1 year up to 10 years) or when the loss is a functional of a general path-dependent SDE. 

  In the context of hedging or optimizing a portfolio of financial instruments to reduce the CVaR, it is shown in \cite{rockafellarb} that it is possible to compute both VaR and CVaR (actually calculate VaR and optimize CVaR) by solving a convex optimization problem with a linear programming approach. It consists in generating loss scenarios and then in introducing constraints in the linear programming problem. Although they address a different problem, this method can be used to compute both VaR and CVaR. The advantage of such a method is that it is possible to estimate both VaR and CVaR simultaneously without assuming that the market prices have a specified distribution (e.g. normal, log-normal, ...). The main drawback is that the dimension (number of constraints) of the linear programming problem to be solved is equal to the number of simulated scenarios. In our approach, we are no limited by the number of generated sample paths used in the procedure.

  The idea to compute both VaR and CVaR with one procedure comes from the fact that they are strongly linked as they appear as the solutions and the value of the same convex optimisation problem (see Proposition 2.1) as pointed out \cite{rockafellarb}. Moreover both the objective function of the minimization problem and its gradient read as an expectation. This leads us to define consistent and asymptotically normal estimators of both quantities as the limit of a global Robbins-Monro (RM) procedure. Consequently, we are no longer constrained by the number of samples paths used in the estimation.

 A significant advantage of this recursive approach, especially in regard to the inversion of the empirical function method is that we only estimate the quantities of interest and not the whole inverse of the distribution function. Furthermore, we do not need to make approximations of the loss or of the convex optimization problem to be solved. Moreover, the implementation of the algorithm is straightforward. However to make it really efficient we need to modify it owing to the fact that VaR and CVaR computation is closely related to the simulation of rare events. That is why as a necessary improvement, we introduce a (recursive and adaptive) variance reduction method based on an importance sampling (IS) paradigm. 

 Let us be a bit more specific. Basically in this kind of problem we are interested in events that are observed with a very low probability (usually less that $5\%$, $1\%$ or even $0.1\%$) so that we obtain few significant replications to update our estimates. Actually, interesting losses are those that exceed the VaR, $i.e.$ the ones that are ``in the tail'' of the loss distribution. Thus in order to compute more accurate estimates of both quantities of interest, it is necessary to generate more samples in the tail of $L$, the area of interest. A general tool used in this situation is IS. 

 The basic principle of IS is to modify the distribution of $L$ by an equivalent change of measure to obtain more ``interesting'' samples that will lead to better estimates of the VaR and CVaR. The main issue of IS is to find a right change of measure (among a parameterized family) that will induce a significant variance reduction. In \cite{glassermanb} and \cite{glassermanc}, a change of measure based on a large deviation upper bound is proposed to estimate the loss probability $\mathbb{P}(L>x)$ for several values of $x$. Then, it is possible to estimate the VaR by interpolating between the estimated loss probabilities.

 Although this approach provides an asymptotically optimal IS distribution, it is strongly based on the fact that the Delta-Gamma approximation holds exactly and relies on the assumption that, conditionally to the past data, market moves are normally distributed. Moreover, as shown in \cite{glwang}, importance sampling estimators based on a large deviations change of measure can have variance that increases with the rarity of the event, and even infinite variance. In \cite{egloff}, the VaR$_{\alpha}$ is estimated by using a quantile estimator based on the inversion of the empirical weighted function and combined with Robbins-Monro (RM) algorithm with repeated projection devised to produce the optimal measure change for IS purpose. This kind of IS algorithm is known to converge toward the optimal importance sampling parameter only after a (long) stabilization phase and provided that the compact sets have been appropriately specified. By contrast, our parameters are optimized by an adaptive unconstrained ($i.e.$ without projections) RM algorithm naturally combined with our VaR-CVaR procedure.

 One major issue that arises when combining the VaR-CVaR algorithm with the recursive IS procedure is to ensure that the IS parameters do move appropriately toward the critical risk area. They may remain stuck at the very beginning of the IS procedure. To circumvent this problem, we make the confidence level slowly increase from a low level (say $50\%$) to $\alpha$ by introducing a deterministic sequence $(\alpha_{n})_{n\geq0}$ of confidence level that converges toward $\alpha$. This kind of incremental threshold increase has been proposed previously \cite{kroese} in a different framework (use of cross entropy in rare event simulation). It speeds up the initialization phase of the IS algorithm and consequently improves the variance reduction. Thus, we can truly experiment asymptotic convergence results in practice. 

\bigskip

 The paper is organized as follows. In the next section, we present some theoretical results about VaR and CVaR. We introduce the VaR-CVaR stochastic algorithm in its first and naive version and study its convergence rate. We also introduce some background about IS using stochastic approximation algorithm. Section 3 is devoted to the design of an optimal procedure using an adaptive variance reduction procedure. We present how it modifies the asymptotic variance of our first CLT. In Section 4 we provide some extensions to the exponential change of measure and to deal with the case of infinite dimensional setting. Section 5 is dedicated to numerical examples. We propose several portfolios of options on several assets in order to challenge the algorithm and display variance reduction factors obtained using the IS procedure. To prevent the freezing of the algorithm during the first iterations of the IS procedure, we also consider a deterministic moving risk level $\alpha_{n}$ which replace $\alpha$ to speed up the initialization phase and improve the reduction of variance. We prove theoretically that modifying in this way the algorithm doesn't change the previous CLT and fasten the convergence.

\medskip
\noindent {\bf Notations:} $\bullet$  $|.|$ will denote the canonical Euclidean norm on $\mathbb{R}^{d}$ and $\left\langle .,.\right\rangle$ will denote the canonical inner product.

\noindent $\bullet$ $\stackrel{\mathcal{L}}{\longrightarrow}$ will denote the convergence in distribution  and  $\stackrel{a.s.}{\longrightarrow}$ will denote the almost sure convergence.

\noindent $\bullet$ $x_{+}:=\max(0,x)$ will denote the positive part function.

\end{section}

%

\begin{section}{VaR, CVaR using stochastic approximation and some background on recursive IS}
\noindent It is rather natural to consider that the loss of the portfolio over the considered time horizon can be written as a function of a structural finite dimensional random vector, $i.e.$ $L=\varphi(X)$, where $X$ is a $\mathbb{R}^{d}$-valued random vector defined on the probability space $\left(\Omega,A,\mathbb{P}\right)$ and $\varphi: \mathbb{R}^{d} \rightarrow \mathbb{R}$ is a Borel function. $\varphi$ is the function representing the composition of the portfolio which remains fixed and $X$ is a structural random vector used to model the market prices over the time interval; therefore we do not need to specify the dynamics of the market prices and only rely on the fact that it is possible to sample from the distribution of $X$. For instance, in a Black-Scholes framework, $X$ is a Gaussian vector and $\varphi$ can be a portfolio of vanilla options. In more sophisticated models or portfolio, $X$ can be a vector of Brownian increments related to the Euler scheme of a diffusion. The VaR at level $\alpha\in\left(0,1\right)$ is the lowest $\alpha$-quantile of the distribution $\varphi(X)$ \textsl{$i.e.$}:

$$
\text{VaR}_{\alpha}(\varphi(X)):=\inf\left\{\xi \mbox{ }|\mbox{ }\mathbb{P}\left(\varphi(X)\leq \xi\right)\geq\alpha\right\}.
$$

\noindent Since $\lim_{\xi \rightarrow +\infty}\mathbb{P}\left(\varphi(X)\leq
\xi\right)=1$, we have $\left\{\xi \mbox{ }|\mbox{ }\mathbb{P}\left(\varphi(X)\leq \xi\right)\geq\alpha\right\}\neq \varnothing  $. Moreover, we have $\lim_{\xi\rightarrow-\infty}\mathbb{P}\left(\varphi(X)\leq
\xi\right)=~0$, which implies that $\left\{\xi \mbox{ }|\mbox{ }\mathbb{P}\left(\varphi(X)\leq \xi\right)\geq\alpha\right\}$ is bounded from below so that the VaR always exists. We assume that the distribution function of $\varphi(X)$ is continuous ($i.e.$ without atoms) so that the VaR is the lowest solution of the equation:
$$
\mathbb{P}\left(\varphi(X)\leq \xi\right)=\alpha.
$$

\noindent Three values of $\alpha$ are commonly considered: 0.95, 0.99, 0.995 so that it is usually close to 1 and the tail of interest has probability $1-\alpha$. If the distribution function is (strictly) increasing, the solution of the above equation is unique, otherwise, there may be more than one solution. In fact, in what follows, we will consider that \emph{any} solution of the previous equation is the VaR. Another risk measure generally used to provide information about the tail of the distribution of $\varphi(X)$ is the \textsl{Conditional Value-at-Risk} (CVaR) (at level $\alpha$). As soon as $\varphi(X) \in L^{1}(\mathbb{P})$, it is defined by:
$$
\text{CVaR}_{\alpha}(\varphi(X)):=\mathbb{E}\left[\varphi(X) | \varphi(X) \geq \text{VaR}_{\alpha}(\varphi(X))\right].
$$

\noindent The CVaR of $\varphi(X)$ is simply the conditional expectation of $\varphi(X)$ given that it lies inside the critical risk area. To capture more information on the conditional distribution of $\varphi(X)$, it seems natural to consider more general risk measures like for example the conditional variance. In a more general framework we can be interested in estimating the $\Psi$-\textsl{Conditional Value at Risk} ($\Psi$-CVaR) (at level $\alpha$) where  $\Psi:\mathbb{R}\rightarrow \mathbb{R}$ is a continuous function. As soon as $\Psi(\varphi(X)) \in L^{1}(\mathbb{P})$, it is defined by:

\begin{equation}{}
\Psi\mbox{-}\textnormal{CVaR}_{\alpha}(\varphi(X)):=\mathbb{E}\left[\Psi(\varphi(X)) | \varphi(X) \geq \text{VaR}_{\alpha}(\varphi(X))\right].
\label{PhiCVaR}
\end{equation}

\noindent When $\Psi \equiv Id$ and $\varphi(X) \in L^{1}(\mathbb{P})$,~\eqref{PhiCVaR} is the regular CVaR of $\varphi(X)$. When $\Psi\equiv x \mapsto x^{2}$, equation~\eqref{PhiCVaR} is but the conditional quadratic norm of $\varphi(X)$.

\medskip

\subsection{Representation of VaR and $\Psi\mbox{-CVaR}$ as expectations}
\noindent The idea to devise a stochastic approximation procedure to compute VaR and CVaR, and more generally the $\Psi$-CVaR, comes from the fact that these two quantities are solutions
of a convex optimization problem whose value function can be represented as an expectation as pointed out by Rockafellar and Uryasev in \cite{rockafellar}.

\begin{prop}
Let $V$ and $V_{\Psi}$ be the functions defined by:
\begin{equation}{}
V(\xi)=\mathbb{E}\left[v(\xi,X)\right] \hspace*{0.2cm} \mbox{ and } \hspace*{0.2cm}
V_{\Psi}(\xi)=\mathbb{E}\left[w(\xi,X)\right]
\label{Vfunctions}
\end{equation}

\noindent where
\begin{equation}{}
v(\xi,x):=\xi + \frac{1}{1-\alpha}(\varphi(X)-\xi)_{+} \hspace*{0.2cm} \mbox{ and } \hspace*{0.2cm}
w(\xi,x):=\xi+\frac{1}{1-\alpha} (\Psi(\varphi(x))-\xi) \mbox{\bf 1}_{\left\{\varphi(x) \geq \xi \right\}}.
\label{vwfunctions}
\end{equation}

\noindent Suppose that the distribution function of $\varphi(X)$ is continuous and that $\varphi(X) \in L^{1}(\mathbb{P})$. Then, the function $V$ is convex, differentiable and the $\textnormal{VaR}_{\alpha}(\varphi(X))$ is any point of the set:

$$\arg\min V= {\left\{\xi \in \mathbb{R}\ | \ V'(\xi)=0 \right\}}={\left\{\xi \mbox{ } | \mbox{ } \mathbb{P}(\varphi(X) \leq \xi)=\alpha\right\}},$$

\noindent where $V'$ is the derivative of $V$ defined for every $\xi \in \mathbb{R}$ by 
\begin{equation}
V'(\xi)=\mathbb{E}\left[\frac{\partial v}{\partial \xi} \left(\xi,X\right)\right].
\label{gradientV}
\end{equation} 

\noindent Furthermore,

$$
\textnormal{CVaR}_{\alpha}(\varphi(X))=\min_{\xi \in \mathbb{R}}V(\xi)
$$

\noindent and, if $\Psi$ is continuous and that $\Psi\left(\varphi(X)\right) \in L^{1}(\mathbb{P})$, for every $\xi^{*}_{\alpha} \in \arg \min V$ ($i.e.$, $\xi^{*}_{\alpha}$ is a $\textnormal{VaR}_{\alpha}(\varphi(X))$)

$$
\Psi\mbox{-}\textnormal{CVaR}_{\alpha}(\varphi(X))=V_{\Psi}(\xi^{*}_{\alpha}).
$$
\end{prop}

\begin{proof}[ {\textbf{Proof.}}]
Since the functions $\xi \mapsto (\varphi(x)-\xi)_{+}$, $x\in \mathbb{R}^{d}$, are convex, the function V is convex. $\mathbb{P}(dw)$-$a.s.$, $\frac{\partial v}{\partial \xi}(\xi,X(w))$ exists at every $\xi \in \mathbb{R}$ and 

$$ 
\mathbb{P}(dw) \mbox{-} a.s., \ \  \left|\frac{\partial v}{\partial \xi}( \xi,X(w))\right|\leq 1\vee\frac{\alpha}{1-\alpha}.
$$

Thanks to Lebesgue Dominated Convergence Theorem, one can interchange differentiation and expectation, so that $V$ is differentiable with derivative
$V'(\xi)=1-\frac{1}{1-\alpha}\mathbb{P}(\varphi(X)>\xi)$ and reaches its absolute minimum at any
$\xi^{*}_{\alpha}$ satisfying $\mathbb{P}(\varphi(X)>\xi^{*}_{\alpha})=1-\alpha$ \textsl{$i.e.$} $\mathbb{P}(\varphi(X)\leq\xi^{*}_{\alpha})=\alpha$. 

\noindent Moreover, it is clear that:
\begin{eqnarray*}
V(\xi^{*}_{\alpha}) &=& \xi^{*}_{\alpha}+ \frac{\mathbb{E}[(\varphi(X)-\xi^{*}_{\alpha})_{+}]}{\mathbb{P}(\varphi(X) > \xi^{*}_{\alpha})} \\
&=& \frac{\xi^{*}_{\alpha}\mathbb{E}[\mbox{\bf 1}_{\varphi(X) > \xi^{*}_{\alpha}}]+\mathbb{E}[(\varphi(X)-\xi^{*}_{\alpha})_{+}]}{\mathbb{P}(\varphi(X) > \xi^{*}_{\alpha})} \\
&=& \mathbb{E}\left[\varphi(X)|\varphi(X)>\xi^{*}_{\alpha}\right]
\end{eqnarray*}
and, in the same way, $V_{\Psi}(\xi^{*}_{\alpha})=\Psi\mbox{-}\textnormal{CVaR}_{\alpha}(\varphi(X))$. This completes the proof.
\end{proof}

\noindent \textbf{Remark:} Actually, one could consider a more general framework by including any risk measure defined by an integral representation with respect to $X$: 
$$
\mathbb{E}[\Lambda(\xi^{*}_{\alpha},X)]
$$ 
where $\Lambda$ is a (computable) Borel function.

\medskip
\subsection{Stochastic gradient and its adaptive companion procedure: a first naive approach} 

\noindent The above representation~\eqref{gradientV} naturally yields a stochastic gradient procedure derived from the convex Lyapunov function $V$ which will (hopefully) converge toward $\xi^{*}_{\alpha}:=\textnormal{VaR}_{\alpha}(\varphi(X))$. Then, a recursive companion procedure based on~\eqref{Vfunctions} can be easily devised having $C^{*}_{\alpha}:=\Psi\text{-}\textnormal{CVaR}_{\alpha}(\varphi(X))$ as target. There is no reason to believe that this first version can do better than the empirical quantile estimate. But, it is a necessary phase in order to understand how our recursive IS algorithm (to be devised further on) can be combined with this first procedure.

\noindent First we set 
\begin{equation}
H_{1}(\xi,x):=\frac{\partial v}{\partial \xi}(\xi,x)=1-\frac{1}{1-\alpha}\mbox{\bf 1}_{\left\{\varphi(x)\geq\xi\right\}},
\label{H1def}
\end{equation}

\noindent so that, 
$$
V'(\xi)=\mathbb{E}\left[H_{1}\left(\xi,X\right)\right].
$$

\noindent Since we are looking for $\xi$ for which $\mathbb{E}\left[H_{1}(\xi,X)\right]=0$, we implement a stochastic gradient descent derived from the Lyapunov function $V$ to approximate $\xi^{*}_{\alpha}:=VaR_{\alpha}(\varphi(X))$, $i.e.$, we use the RM algorithm:
\begin{equation}{}
    \xi_{n}= \xi_{n-1}-\gamma_{n}H_{1}(\xi_{n-1},X_{n}), n \geq 1, \ \ \xi_{0}\in L^{1}(\mathbb{P}),
\label{RMVaR}
\end{equation}

\noindent where $(X_{n})_{n\geq1}$ is an i.i.d. sequence of random variables with the same distribution as $X$, independent of $\xi_{0}$, with $\mathbb{E}[|\xi_{0}|]<+\infty$ and $(\gamma_{n})_{n \geq 1}$ is a deterministic step sequence (decreasing to $0$) satisfying:
\begin{hyp}
\sum_{n\geq1} \gamma_{n}=+\infty \ \ \mbox{ and } \ \ \sum_{n\geq1} \gamma_{n}^{2}<+\infty.
\label{stepC}
\end{hyp}

\noindent In order to derive the $a.s.$ convergence of~\eqref{RMVaR} we introduce the following additional assumption on the distributions of $\varphi(X)$ and $\Psi(\varphi(X))$. Let $a>0$,
\begin{hyp1}
\varphi(X) \mbox{ has a continuous distribution function and }\Psi(\varphi(X)) \in L^{2a}(\mathbb{P}).
\end{hyp1}

\noindent Actually, Equation~\eqref{RMVaR} can be seen either as a regular RM procedure with mean function $V'$ since it is increasing (see e.g. \cite{duflob} p.50 and p.66) or as a recursive gradient descent procedure derived from the Lyapunov function $V$. Both settings yield the $a.s.$ convergence toward its target $\xi^{*}_{\alpha}$. To establish the $a.s.$ convergence of $(\xi_{n})_{n\geq1}$ (and of our different RM algorithms), we will rely on the following theorem. For a proof of this slight extension of Robbins-Monro Theorem and of the $a.s.$ convergence of $(\xi_{n})_{n\geq1}$ (under assumptions $(A1)$ and $(A2)_{1}$), we refer to~\cite{frikha}. 

\begin{thm} \textsl{(Robbins-Monro Theorem (variant)).} Let $H:\mathbb{R}^{q}\times\mathbb{R}^{d}\rightarrow \mathbb{R}^{d}$ be a Borel function and $X$ be an $\mathbb{R}^{d}$-valued random vector such that $\mathbb{E}[|H(z,X)|]< \infty$ for every $z \in \mathbb{R}^{d}$. Then set

$$
\forall z \in \mathbb{R}^{d}, \hspace*{0.3cm} h(z)=\mathbb{E}[H(z,X)].
$$ 

\noindent Suppose that the function $h$ is continuous and that $\mathcal{T}^{*}:={\left\{h=0\right\}}$ satisfies 

\begin{equation}
\forall z \in \mathbb{R}^{d}\mbox{ \textbackslash }\mathcal{T}^{*}, \forall z^{*} \in \mathcal{T}^{*}, \hspace*{0.3cm} \left\langle z-z^{*}, h(z)\right\rangle>0.
\label{MeanRev}
\end{equation}

\noindent Let $(\gamma_{n})_{n \geq 1}$ be a deterministic step sequence satisfying condition (A1). Suppose that 
\begin{equation} \forall z \in \mathbb{R}^{d},\ \ \mathbb{E}[|H(z,X)|^{2}] \leq C(1+|z|^{2})
\label{LinearGrowthAssump} 
\end{equation} 

\noindent (which implies that $|h(z)| \leq C'(1+|z|)$). 

\noindent Let $(X_{n})_{n \geq 1}$ be an i.i.d. sequence of random vectors having the distribution of $X$, let $z_{0}$ be a random vector independent of $(X_{n})_{n \geq 1}$ satisfying $\mathbb{E}[|z_{0}|]< \infty$, all defined on the same probability space $(\Omega, A, \mathbb{P})$. Let $\mathcal{F}_{n}:=\sigma(z_{0},X_{1},...,X_{n})$ and let $(r_{n})_{n\geq1}$ be an $\mathcal{F}_{n}$-measurable remainder sequence satisfying
\begin{equation}
\sum_{n} \gamma_{n} |r_{n}|^{2}< \infty.
\label{condremaindersequ}
\end{equation}

\noindent Then, the recursive  procedure defined for $n\geq1$ by 
$$
Z_{n}=Z_{n-1}-\gamma_{n}H(Z_{n-1},X_{n})+\gamma_{n}r_{n},
$$ 

\noindent satisfies: 
$$
\exists \ z_{\infty}, \hspace*{0.1cm} \mbox{ \textsl{ such that } } Z_{n}\stackrel{a.s.}{\longrightarrow} z_{\infty} \mbox{ and } z_{\infty} \in \mathcal{T}^{*} \mbox{ }a.s.
$$

\noindent The convergence also holds in $L^{p}(\mathbb{P}),p \in (0,2)$, where $L^{p}(\mathbb{P})$ denotes the set of all random vectors defined on $\left(\Omega,A,\mathbb{P}\right)$ such that $\mathbb{E}[\left|X\right|^{p}]^{\frac{1}{p}}<\infty$.
\end{thm}

\noindent \textbf{Remark:} It is in fact a slight variant (see e.g. \cite{frikha}) of the regular RM Theorem since $Z_{n}$ converges to a random vector having its value in the set $\left\{{h=0}\right\}$ even if $\left\{{h=0}\right\}$ is not reduced to a singleton or a finite set. The remainder sequence in the above theorem plays a crucial role when we will (slightly) modify the first IS procedure to improve its efficiency. \newline

The second step concerns procedure for the numerical computation of the $\Psi\mbox{-}\textnormal{CVaR}_{\alpha}$. A naive idea is to compute the function $V_{\Psi}$ at the point $\xi^{*}_{\alpha}$:

$$
\Psi\mbox{-}\textnormal{CVaR}_{\alpha}=V_{\Psi}(\xi^{*}_{\alpha})=\mathbb{E}[w(\xi^{*}_{\alpha},X)]
$$

\noindent using a regular Monte Carlo simulation,
\begin{eqnarray}
\frac{1}{n}\sum_{k=0}^{n-1}w(\xi^{*}_{\alpha},X_{k+1}).
\label{RMCVaR1}
\end{eqnarray}

\noindent However, we first need to get from~\eqref{RMVaR} a good approximate of $\xi^{*}_{\alpha}$ and subsequently to use another sample of the distribution $X$. A natural idea is to devise an adaptive \textsl{companion procedure} of the above quantile search algorithm by replacing $\xi^{*}_{\alpha}$ in~\eqref{RMCVaR1} by its approximation at step $k$, namely
\begin{eqnarray}
C_{n}=\frac{1}{n}\sum_{k=0}^{n-1}w(\xi_{k},X_{k+1}), \ n\geq1, \ C_{0}=0.
\label{MoyEmpCVaR}
\end{eqnarray}

\noindent Hence, $(C_{n})_{n\geq 0}$ is the sequence of empirical means of the non i.i.d. sequence $(w(\xi_{k},X_{k+1}))_{k\geq1}$, which can be written recursively:
\begin{equation}{}
    C_{n}=C_{n-1}-\frac{1}{n}H_{2}\left(\xi_{n-1},C_{n-1},X_{n}\right),\ n\geq1,
\label{RMCVaR2}
\end{equation}

\noindent where $H_{2}\left(\xi,c,x\right):=c-w(\xi,x).$ 

 At this stage, we are facing two procedures $(\xi_{n},C_{n})$ with different steps. This may appear not very consistent or at least natural. A second modification to the original Monte Carlo procedure~\eqref{RMCVaR2} consists in considering a general step $\beta_{n}$ satisfying condition (A1) instead of $\frac{1}{n}$ (with in mind the possibility to set $\beta_{n}=\gamma_{n}$ eventually). This leads to:
\begin{equation}{}
    C_{n}=C_{n-1}-\beta_{n}H_{2}\left(\xi_{n-1},C_{n-1},X_{n}\right),\ n\geq1.
\label{RMCVaR}
\end{equation}

\noindent In order to prove the $a.s.$ convergence of $(C_{n})_{n\geq1}$ toward $C^{*}_{\alpha}$, we set for convenience $\beta_{0}:= \sup_{n\geq1} \beta_{n} +1$. Then, one defines recursively a sequence $(\Delta_{n})_{n\geq1}$ by 

$$
\Delta_{n+1} = \Delta_{n} \frac{\beta_{n+1}}{\beta_{n}} \frac{\beta_{0}}{\beta_{0}-\beta_{n+1}}, \mbox{ } n \geq 0, \ \Delta_{0}=1.
$$ 

\noindent Elementary computations show by induction that 
\begin{equation}
\beta_{n}=\beta_{0}\frac{\Delta_{n}}{S_{n}}, \mbox{ } n \geq 0, \mbox{ with } S_{n}=\sum_{k=0}^{n} \Delta_{k}. 
\label{FormulePas}
\end{equation}

\noindent Furthermore, it follows from~\eqref{FormulePas} that for every $n \geq 1$ 
$$
\log(S_{n})-\log(S_{n-1})=-\log\left(1-\frac{\Delta_{n}}{S_{n}}\right) \geq \frac{\Delta_{n}}{S_{n}}= \frac{\beta_{n}}{\beta_{0}}.
$$
 
\noindent Consequently, $$\log(S_{n}) \geq \frac{1}{\beta_{0}} \sum_{k=1}^{n} \beta_{k}$$ which implies that $\lim_{n} S_{n} =+\infty.$ 

\noindent Now using~\eqref{RMCVaR} and~\eqref{FormulePas}, one gets for every $n \geq 1$ 
$$ 
S_{n}C_{n}=S_{n-1}C_{n-1}+\Delta_{n}\left(\Delta N_{n+1}+V_{\Psi}(\xi_{n})\right)
$$

\noindent where, $\Delta N_{n}:=w(\xi_{n-1},X_{n})-V_{\Psi}(\xi_{n-1})$, $n \geq 1$, define a martingale increments sequence with respect to the natural filtration of the algorithm $\mathcal{F}_{n}:=\sigma(\xi_{0},X_{1},\cdots,X_{n}),\mbox{ }n\geq0$. Consequently,

$$C_{n}=\frac{1}{S_{n}}\left(\sum^{n-1}_{k=0} \Delta_{k+1} \Delta N_{k+1} +\sum^{n-1}_{k=0} \Delta_{k+1}V_{\Psi}(\xi_{k})\right).$$ 

\noindent The second term in the right hand side of the above equality converges to $V_{\Psi}(\xi^{*}_{\alpha})=\Psi\mbox{-}\textnormal{CVaR}_{\alpha}(\varphi(X))$ owing to the continuity of $V_{\Psi}$ at $\xi^{*}_{\alpha}$ and Cesaro's Lemma. 

\noindent The convergence to 0 of the first term will follow from the $a.s.$ convergence of the series
 
$$
N_{n}^{\beta}:=\sum_{k=1}^{n} \beta_{k} \Delta N_{k}, \mbox { } n \geq 1
$$
 
\noindent by the Kronecker Lemma since $\beta_{n}=\beta_{0} \frac{\Delta_{n}}{S_{n}}$. The sequence $(N_{n}^{\beta})_{n\geq1}$ is an $\mathcal{F}_{n}$-martingale since the $\Delta N_{k}$'s are martingale increments and
$$
\mathbb{E}\left[(\Delta N_{n})^{2} | \mathcal{F}_{n-1}\right] \leq \frac{1}{(1-\alpha)^{2}} \mathbb{E}\left[\left(\Psi\left(\varphi\left(X\right)\right)-\xi\right)^{2}\right]_{|\xi=\xi_{n-1}}.
$$

\noindent Assumption $(A2)_{1}$ and the $a.s.$ convergence of $\xi_{k}$ toward $\xi^{*}_{\alpha}$ imply that

$$
\sup_{n\geq1} \mathbb{E}[(\Delta N_{n})^{2}|\mathcal{F}_{n-1}]< \infty \mbox{  } \ \ a.s. 
$$

\noindent Consequently, assumption (A1) implies 

$$
\langle N^{\beta}\rangle_{\infty}=\sum_{n \geq 1} \beta_{n}^{2}\mathbb{E}[(\Delta N_{n})^{2}|\mathcal{F}_{n-1}]< \infty
$$

\noindent which in term yields the $a.s.$ convergence of $(N_{n}^{\beta})_{n \geq 1}$, so that $C_{n}\stackrel{a.s.}{\longrightarrow} \Psi\mbox{-}CVaR_{\alpha}\left(\varphi(X)\right).$

\noindent The resulting algorithm reads as for $n\geq1$:
\begin{equation}{}
\left\{
\begin{array}{l}
    \xi_{n}= \xi_{n-1}-\gamma_{n}H_{1}\left(\xi_{n-1},X_{n}\right), \ \ \xi_{0}\in L^{1}(\mathbb{P}),  \vspace*{0.2cm} \\
    C_{n}=C_{n-1}-\beta_{n}H_{2}\left(\xi_{n-1},C_{n-1},X_{n}\right), \ \ C_{0}=0,  
\end{array}
\label{RMVaRCVaR}
\right.
\end{equation}  

\noindent and converges under $(A1)$ and $(A2)_{1}$.

The question of the joint weak convergence rate of $\left(\xi_{n},C_{n}\right)$ is not trivial owing to the coupling of the two procedures. The case of two different step scales refers to the general framework of two-time-scale stochastic approximation algorithms. Several results have been established by Borkar in \cite{borkar}, Konda and Tsitsiklis in \cite{konda} but the more relevant in our case are those of Mokkadem and Pelletier in \cite{mokkadem}. 
\noindent The weak convergence rate of $(\xi_{n})_{n\geq1}$ is ruled by the CLT for ``regular'' (single-time scale) stochastic approximation algorithms (we refer to Kushner and Clark in \cite{kushnera}, M\'{e}tivier and Priouret in \cite{benv}, Duflo in \cite{duflob} among others). In order to achieve the best asymptotic rate of convergence, one ought to set $\gamma_{n}=\frac{\gamma_{0}}{n}$ where the choice of $\gamma_{0}$ depends on the value of the density $f_{\varphi(X)}$ of $\varphi(X)$ at $\xi^{*}_{\alpha}$, which is unknown. To circumvent the difficulties induced by the specification of $\gamma_{0}$, which are classical in this field, we are led to modify again our algorithm by introducing the averaging principle independently introduced by Ruppert~\cite{ruppert} and Polyak~\cite{polyak} and then widely investigated by several authors. It works both with two-time or single-time scale steps and leads to asymptotically efficient procedures, $i.e.$, satisfying a CLT at the optimal rate $\sqrt{n}$ and minimal variance (see also \cite{mokkadem}). See also a variant based on a gliding window developed in \cite{lelong}. Our numerical examples indicate that the averaged one-time-scale procedure provides less variance during the first iterations than the averaged procedure of the two-time-scale algorithm. Finally, we set $\gamma_{n}\equiv \beta_{n}$ in~\eqref{RMVaRCVaR} so that, the VaR-CVaR algorithm can be written in a more synthetic way by setting $Z_{n}=(\xi_{n},C_{n})$ and for $n\geq1$: 

\begin{equation}
Z_{n}=Z_{n-1}-\gamma_{n}H(Z_{n-1},X_{n}), \ Z_{0}=\left(\xi_{0},C_{0}\right), \ \xi_{0} \in L^{1}(\mathbb{P}),
\label{VaRCVaRequalsteps}
\end{equation}

\noindent where $H(z,x):=(H_{1}(\xi,x),H_{2}(\xi,C,x))$.
\noindent Throughout the rest of this section, we assume that the distribution $\varphi(X)$ has a positive probability density $f_{\varphi(X)}$ on its support. As a consequence the $VaR_{\alpha}(\varphi(X))$ is unique so that the procedure algorithm $Z_n$ converges $a.s.$ to its single target $\left(\textnormal{VaR}_{\alpha}(\varphi(X)), \Psi\mbox{-}\textnormal{CVaR}_\alpha(\varphi(X))\right)$. Thus, the Cesaro mean of the procedure 
$$
\bar Z_n :=\frac {Z_0+\cdots+Z_{n-1}}{n}, \ \ n\geq 1,
$$
 
\noindent where $Z_{n}$ is defined by~\eqref{VaRCVaRequalsteps}, converges $a.s.$ to the same target. The Ruppert and Polyak's Averaging Principle says that an appropriate choice of the step yields for free the smallest possible asymptotic variance. We recall below this result (following a version established in \cite{duflob}, see \cite{duflob} (p.169) for a proof).

\begin{thm}{(Ruppert and Polyak's Averaging Principle)}
Suppose that the $\mathbb{R}^{d}$-sequence $(Z_{n})_{n\geq 0}$ is defined recursively by
$$
Z_{n}=Z_{n-1}-\gamma_{n}\left(h(Z_{n-1})+\epsilon_{n}+r_{n}\right)
$$

\noindent where $h$ is a Borel function. Let $\mathbb{F}:=(\mathcal{F}_{n})_{n\geq0}$ be the natural filtration of the algorithm, $i.e.$ such that the sequence $(\epsilon_{n})_{n \geq1}$ and $(r_{n})_{n \geq1}$ is $\mathbb{F}-$adapted. Suppose that $h$ is $\mathcal{C}^{1}$ in the neighborhood of a zero $z^{*}$ of $h$ and that $M=Dh(z^{*})$ is  a uniformly repulsive matrix (all its eigenvalues have positive real parts) and that $(\epsilon_{n})_{n \geq 1}$ satisfies
\begin{equation}
\exists \  C >0, \mbox{ such that a.s. }\left\{ \begin{array}{l}

(i) \mbox{  }\mathbb{E}[\epsilon_{n+1}|\mathcal{F}_{n}] \mbox{\bf 1}_{\left\{||Z_{n}-z^{*}|| \leq C \right\}}=0, \vspace*{0.2cm}  \\ 

(ii) \mbox{  } \exists b >2, \mbox{ } \sup_{n} \mathbb{E}[||\epsilon_{n+1}||^{b}|\mathcal{F}_{n}] \mbox{ \bf 1}_{\left\{||Z_{n}-z^{*}|| \leq C \right\}} < +\infty, \vspace*{0.2cm} \\ 

(iii) \ \mathbb{E}\left[(\gamma_{n-1})^{-1}\left|r_{n}\right|^{2}\mbox{ \bf 1}_{\left\{||Z_{n}-z^{*}||\leq C\right\}}\right] \rightarrow 0, \vspace*{0.2cm} \\ 

(iv) \ \exists \ \Gamma \in \mathcal{S}^{+}(d,\mathbb{R}) \mbox{  such that }\mathbb{E}\left[\epsilon_{n+1}\epsilon_{n+1}^{T}|\mathcal{F}_{n}\right]\stackrel{a.s.}{\longrightarrow} \Gamma .  

\end{array}
\right.
\label{hypii}
\end{equation}

\noindent Set $\gamma_{n}=\frac{\gamma_{1}}{n^{a}}$ with $\frac{1}{2}< a < 1$,
and 

$$
\bar{Z}_{n+1}:=\frac{Z_{0}+...+Z_{n}}{n+1}=\bar{Z}_{n}-\frac{1}{n+1}(\bar{Z}_{n}-Z_{n}),
\mbox{ } n \geq 0.
$$ 

\noindent Then, on the set of convergence
${\left\{Z_{n}\rightarrow z^{*}\right\}}$:

\begin{equation*}
\sqrt{n}\left(\bar{Z}_{n}-z^{*}\right)\stackrel{\mathcal{L}}{\rightarrow}\mathcal{N}\left(0,M^{-1}\Gamma(M^{-1})^{T}\right)\hspace*{.5cm} \mbox{ as } n \rightarrow +\infty,
\label{RuppertPolyakCLT}
\end{equation*}

\noindent where $(M^{-1})^{T}$ denotes the transpose of the matrix $M^{-1}$.
\end{thm}

\noindent To apply this theorem to our framework we are led to compute the Cesaro means of both components, namely for $n\geq1$
\begin{eqnarray} {}\left\{
\begin{array}{l}
 \overline{\xi}_{n}:=\frac{1}{n} \sum_{k=1}^{n} \xi_{k}=\overline{\xi}_{n-1}-\frac{1}{n}(\overline{\xi}_{n-1}-\xi_{n}), \vspace*{0.2cm} \\
 \overline{C}_{n}:=\frac{1}{n} \sum_{k=1}^{n} C_{k}=\overline{C}_{n-1}-\frac{1}{n}(\overline{C}_{n-1}-C_{n}), 
 \end{array}
\label{AveRMVaRCVaR}
\right.
\end{eqnarray}

\noindent where $(\xi_{k},C_{k})$, $k\geq0$ is defined by~\eqref{VaRCVaRequalsteps}. In the following theorem, we provide the convergence rate of the couple $\bar{Z}_{n}:=\left(\overline{\xi}_{n},\overline{C}_{n}\right)$.

\begin{thm}{(Convergence rate of the VaR-CVaR procedure)}. Suppose $(A2)_{a}$ holds for some $a>1$, the density function $f_{\varphi(X)}$ of
$\varphi(X)$ is continuous, strictly positive at $\xi^{*}_{\alpha}$. If the step sequence is $\gamma_{n}=\frac{\gamma_{1}}{n^{a}}$ with $ \frac{1}{2} < a <
1$ and $\gamma_{1}>0$ then

\begin{eqnarray*}
\sqrt{n}\left(\bar{Z}_{n}-z^{*}\right)\stackrel{\mathcal{L}}{\longrightarrow}
\mathcal{N}\left(0,\Sigma\right) \hspace*{.5cm} \mbox{ as } n \rightarrow +\infty 
\label{CLTVaRCVaR} 
\end{eqnarray*}

\noindent where the asymptotic covariance matrix $\Sigma$ is given by
\begin{equation}
\begin{pmatrix}
\frac{\alpha(1-\alpha)}{f_{\varphi(X)}^{2}(\xi^{*}_{\alpha})}& \frac{\alpha}{(1-\alpha)f_{\varphi(X)}(\xi^{*}_{\alpha})} \mathbb{E}\left[ \left(\Psi(\varphi(X))-\xi^{*}_{\alpha}\right) \mbox{\bf 1}_{\left\{ \varphi(X) \geq \xi^{*}_{\alpha}\right\} } \right]  \vspace*{0.2cm} \\
\frac{\alpha}{(1-\alpha)f_{\varphi(X)}(\xi^{*}_{\alpha})}\mathbb{E}\left[ \left(\Psi(\varphi(X))-\xi^{*}_{\alpha}\right) \mbox{\bf 1}_{\left\{\varphi(X) \geq \xi^{*}_{\alpha}\right\}}\right] & \frac{1}{(1-\alpha)^{2}}\mbox{\textnormal{Var}}\left(\left(\Psi(\varphi(X))-\xi^{*}_{\alpha}\right) \mbox{\bf 1}_{\left\{ \varphi(X) \geq \xi^{*}_{\alpha}\right\} }\right)
\end{pmatrix}.
\label{varasympt}
\end{equation}
\end{thm}

\begin{proof}[\bf Proof.]
First, the procedure~\eqref{VaRCVaRequalsteps} can be written as for $n\geq1$
\begin{eqnarray}
Z_{n}=Z_{n-1}-\gamma_{n}\left(h(Z_{n-1})+\epsilon_{n}\right), \ \ Z_{0}=\left(\xi_{0},C_{0}\right), \ \ \xi_{0} \in L^{1}(\mathbb{P}),
\label{RMstd1}
\end{eqnarray} 

\noindent where $h(z):=\mathbb{E}[H(z,X)]=\left(1-\frac{1}{1-\alpha}\mathbb{P}\left(\varphi(X) \geq \xi\right), C-\mathbb{E}[w(\xi,X)]\right)$ and $\epsilon_{n}:=(\Delta M_{n}, \Delta N_{n})$, \ $n\geq1$, denotes the $\mathcal{F}_{n}$-adapted martingale increment sequence with

$$
\Delta M_{n}:=\frac{1}{1-\alpha}\left(\mathbb{P}\left(\varphi(X)\geq\xi\right)_{| \xi=\xi_{n-1}}-\mbox{ \bf 1}_{ \left\{\varphi(X_{n}) \geq \xi_{n-1} \right\} }\right).
$$ 
 
\noindent Owing to Assumption $(A2)_{a}$ and Lebesgue's differentiation Theorem, one can interchange expectation and derivation, so that the function $h$ is differentiable at $z^{*}=(\xi^{*}_{\alpha},C^{*}_{\alpha})$ and
\begin{eqnarray}
h'(z^{*})=M:=\begin{pmatrix}
\frac{1}{1-\alpha} f_{\varphi(X)}(\xi^{*}_{\alpha})& 0 \vspace*{0.2cm} \\
\mathbb{E}\left[\left(\frac{\partial}{\partial \xi} w(\xi,X)\right)_{|\xi=\xi^{*}_{\alpha}}\right]& 1 \\
\end{pmatrix}.
\label{Jacobh}
\end{eqnarray}

\noindent Now, $\mathbb{E}\left[\left(\frac{\partial}{\partial \xi} w(\xi,X)\right)_{|\xi=\xi^{*}_{\alpha}}\right]= \left(1- \frac{1}{1-\alpha} \mathbb{P}(\varphi(X) \geq \xi^{*}_{\alpha})\right)=0$, so that, $ M=\begin{pmatrix}
\frac{1}{1-\alpha} f_{\varphi(X)}(\xi^{*}_{\alpha})& 0 \vspace*{0.2cm}\\
0& 1 \\ \end{pmatrix}$ is diagonal. Since $f_{\varphi(X)}$ is continuous at $\xi^{*}_{\alpha}$, $h$ is $\mathcal{C}^{1}$ in the neighborhood of $z^{*}$. 

\noindent To apply Theorem 2.3, we need to check assumptions $(i)\mbox{-}(iv)$ of~\eqref{hypii}.

\noindent Let $A>0$. First note that 

$$
\mathbb{E}\left[ \Delta M_{n+1}^{2a} | \mathcal{F}_{n} \right] \mbox{\bf 1}_{\left\{|Z_{n}-z^{*}| \leq A\right\}} \leq \left(\frac{1}{1-\alpha}\right)^{2a} 2^{2a} < + \infty.
$$

\noindent Thanks to Assumption $(A2)_{a}$, there exists $C_{\alpha,\Psi}>0$ such that
$$
\mathbb{E}\left[ \Delta N_{n+1}^{2a} | \mathcal{F}_{n} \right] \mbox{ \bf 1}_{\left\{||Z_{n}-z^{*}|| \leq A\right\}} \leq C_{\alpha,\Psi} \left(1+\xi_{n}^{2a}\right)  \mbox{ \bf 1}_{\left\{||Z_{n}-z^{*}|| \leq A\right\}} <  + \infty.
$$

\noindent Consequently, $(ii)$ of~\eqref{hypii} holds true with $b=2a>2$ since 

$$ 
\sup_{n \geq 0} \mathbb{E}\left[ | \epsilon_{n+1} |^{2a} | \mathcal{F}_{n} \right] \mbox{\bf 1}_{\left\{|Z_{n}-z^{*}| \leq A\right\}} < +\infty.
$$ 

\noindent It remains to check $(iv)$ for some positive definite symmetric matrix $\Gamma$. The dominated convergence theorem implies that
\begin{eqnarray*}
 \mathbb{E}\left[\left(\epsilon_{n+1}\epsilon_{n+1}^{T}\right)_{1,1}| \mathcal{F}_{n}\right]    &   =   &    \left(\frac{1}{1-\alpha}\right)^{2} \left(\mathbb{E}\left[\mbox{ \bf{1}}_{\left\{\varphi(X)\geq\xi\right\}}\right]_{|\xi=\xi_{n}}-\mathbb{E}\left[ \mbox{ \bf{ 1}}_{\left\{\varphi(X) \geq \xi\right\}}\right]_{|\xi=\xi_{n}}^{2}\right)  \\
 &  \stackrel{a.s.}{\longrightarrow} & \frac{\alpha}{1-\alpha},     \vspace*{0.2cm}    \\
 \mathbb{E}\left[\left(\epsilon_{n+1}\epsilon_{n+1}^{T}\right)_{1,2}| \mathcal{F}_{n}\right]   &    =    &    \mathbb{E}\left[\left(\epsilon_{n+1}\epsilon_{n+1}^{T}\right)_{2,1}| \mathcal{F}_{n}\right]     \vspace*{0.2cm}   \\
&    =    & \left(\frac{1}{1-\alpha}\right)^{2} \mathbb{E}\left[\left(\Psi(\varphi(X))-\xi\right) \mbox{ \bf 1}_{\left\{\varphi(X) \geq \xi \right\} }\right]_{| \xi=\xi_{n}}  \\
&         &  \hspace*{4cm} \times \left(1-\mathbb{E}\left[\mbox{ \bf 1}_{\left\{\varphi(X) \geq \xi\right\}}\right]_{|\xi=\xi_{n}}\right)    \vspace*{0.2cm}  \\
&    \stackrel{a.s.}{\longrightarrow}    &    \frac{\alpha}{\left(1-\alpha\right)^{2}}  \mathbb{E}\left[\left(\Psi(\varphi(X))-\xi^{*}_{\alpha}\right) \mbox{ \bf 1}_{\left\{\varphi(X) \geq \xi^{*}_{\alpha}\right\}}\right],   \vspace*{0.2cm}   \\
 \end{eqnarray*} 
 
 \begin{eqnarray*}
 \mathbb{E}\left[\left(\epsilon_{n+1}\epsilon_{n+1}^{T}\right)_{2,2}| \mathcal{F}_{n}\right]  &     =     &     \mathbb{E}\left[\left(\Delta N_{n+1}\right)^{2}| \mathcal{F}_{n}\right]  \vspace*{0.2cm}  \\
&     =    & \frac{1}{\left(1-\alpha\right)^{2}}\left(\mathbb{E} \left[\left(\Psi(\varphi(X_{n+1}))-\xi\right) \mbox{ \bf 1}_{\left\{\varphi(X_{n+1}) \geq \xi \right\}}| \mathcal{F}_{n}\right]_{|\xi=\xi_{n}} \right. \vspace*{0.2cm}   \\
&   \mbox{ }  & \left. \hspace*{4cm} -\mathbb{E}\left[\left(\Psi(\varphi(X))-\xi\right) \mbox{ \bf 1}_{\left\{\varphi(X) \geq \xi \right\}}\right]_{|\xi=\xi_{n}}^{2} \right)  \vspace*{0.2cm}  \\
&   \stackrel{a.s.}{\longrightarrow}    & \frac{1}{\left(1-\alpha\right)^{2}}\left(\mathbb{E} \left[\left(\Psi(\varphi(X))-\xi^{*}_{\alpha}\right)^{2} \mbox{ \bf 1}_{\left\{\varphi(X) \geq \xi^{*}_{\alpha}\right\}}\right]\right.  \vspace*{0.2cm} \\ 
&    \mbox{ }    & \left.\hspace*{4cm} -\mathbb{E}\left[\left(\Psi(\varphi(X))-\xi^{*}_{\alpha}\right)\mbox{ \bf 1}_{\left\{\varphi(X) \geq \xi^{*}_{\alpha}\right\}}\right]^{2}\right)   \vspace*{0.2cm}  \\
&    =    & \frac{1}{\left(1-\alpha\right)^{2}} \text{Var}\left(\left(\Psi(\varphi(X))-\xi^{*}_{\alpha}\right)\mbox{ \bf 1}_{\left\{\varphi(X) \geq \xi^{*}_{\alpha}\right\}} \right).\\
\end{eqnarray*} 

\noindent Using the continuity  of both functions $\xi \mapsto \mathbb{E}\left[\left(\Psi(\varphi(X))-\xi\right) \mbox{ \bf 1}_{\left\{\varphi(X) \geq \xi \right\}}\right]$ and 

\noindent $\xi \mapsto \mathbb{E}\left[\left(\Psi(\varphi(X))-\xi\right)^{2} \mbox{ \bf 1}_{\left\{\varphi(X) \geq \xi \right\}}\right]$ at $\xi^{*}_{\alpha}$, which follows from the continuity of $\Psi$ and of the distribution function of $\varphi(X)$, finally yields the $a.s.$ convergence of $\mathbb{E}\left[\epsilon_{n+1}\epsilon_{n+1}^{T}|\mathcal{F}_{n}\right]$ toward

$$ 
\Gamma=\begin{pmatrix}
\frac{\alpha}{1-\alpha} & \frac{\alpha}{\left(1-\alpha\right)^{2}}   \mathbb{E}\left[\left(\Psi(\varphi(X))-\xi^{*}_{\alpha}\right) \mbox{ \bf 1}_{\left\{\varphi(X) \geq \xi^{*}_{\alpha}\right\}}\right] \\
\frac{\alpha}{\left(1-\alpha\right)^{2}}  \mathbb{E}\left[\left(\Psi(\varphi(X))-\xi^{*}_{\alpha}\right) \mbox{ \bf 1}_{\left\{\varphi(X) \geq \xi^{*}_{\alpha}\right\}}\right]& \frac{1}{\left( 1-\alpha \right)^{2}} \text{Var}\left(\left(\Psi(\varphi(X))-\xi^{*}_{\alpha}\right) \mbox{ \bf 1}_{\left\{ \varphi(X) \geq \xi^{*}_{\alpha}\right\}} \right) \\
\end{pmatrix}.
$$

\noindent If $\gamma_{n}=\frac{\gamma_{1}}{n^{a}}$ with $\gamma_{1}>0$ and $\frac{1}{2}<a<1$, Ruppert-Polyak's Theorem implies that

$$
\sqrt{n}\left(\bar{Z}_{n}-z^{*}\right)\stackrel{\mathcal{L}}{\longrightarrow}\mathcal{N}\left(0,\Sigma\right)
$$
 
\noindent where $\Sigma = M^{-1}\Gamma \left(M^{-1}\right)^{T}$ is given by~\eqref{varasympt}. This completes the proof.
\end{proof}

\noindent \textbf{Remarks:} $\bullet$ It is possible to replace $w(\xi,x)$ in~\eqref{RMCVaR} and~\eqref{RMVaRCVaR} by $\tilde{w}(\xi,x)=\frac{1}{1-\alpha} \Psi(\varphi(x)) \mbox{\bf 1}_{\left\{\varphi(x) \geq \xi\right\}}$ since $C^{*}_{\alpha}=\mathbb{E}\left[\tilde{w}\left(\xi^{*}_{\alpha},X\right)\right]$. Thus, we only have to change also the martingale increment sequence $(\Delta N_{n})_{n \geq 1}$ by $\left(\Delta \widetilde{N}_{n}\right)_{n\geq1}$ defined by 
$$
\Delta \widetilde{N}_{n}:=\frac{1}{1-\alpha}\left(\mathbb{E}\left[\Psi(\varphi(X)) \mbox{\bf 1}_{\left\{\varphi(X)\geq \xi\right\}}\right]_{|\xi=\xi_{n-1}}-\Psi(\varphi(X_{n}))\mbox{\bf 1}_{\left\{\varphi(X_{n}) \geq \xi_{n-1}\right\}}\right).
$$

\noindent This provides another procedure $\tilde{C}_{n}$ for the computation of the $\Psi\mbox{-}\textnormal{CVaR}_{\alpha}$ which satisfies a Gaussian CLT with the same asymptotic covariance matrix.

\noindent $\bullet$ The quantile estimate based on the inversion of the empirical distribution function satisfies a Gaussian CLT with the same asymptotic covariance matrix than the one of the procedure $\overline{\xi}_{n}$, see for example~\cite{serfling} p.75. Obviously, there is no reason to believe that this first version can do better than the empirical quantile estimate. However, our quantile estimate has the advantage to be recursive: it naturally combines with a recursive IS algorithm in an adaptive way. In terms of computational complexity, once $N$ loss samples have been generated, the behaviour of the inversion of the empirical distribution function method needs a sorting algorithm: good behaviour is $\mathcal{O}\left(N\log(N)\right)$ element comparisons to sort the list of loss samples. Whereas the behaviour of the recursive quantile algorithm is $\mathcal{O}\left(N\right)$. 

\noindent  $\bullet$ One shows that if we choose $\beta_{n}=\frac{1}{n}$, $n \geq 1$ and $\gamma_{n}=\frac{1}{n^{a}}$ with $\frac{1}{2}< a <1$ in~\eqref{RMVaRCVaR}, the resulting two-time scale procedure satisfies a Gaussian CLT with the same asymptotic covariance matrix $\Gamma$ (at rates $\sqrt{\gamma_{n}^{-1}}$ and $\sqrt{n}$). However, by averaging the first component $\xi_{n}$, the resulting procedure becomes asymptotically efficient ($i.e.$ rate $\sqrt{n}$).

\begin{prop}{(Estimation of variance and confidence interval)} 
For every $n\geq 1$, set 
\begin{eqnarray*}
\sigma_{n}^{2} &:= &\frac{1}{(1-\alpha)^{2}}\left( \frac{1}{n} \sum_{k=1}^{n} \left(\Psi(\varphi(X_{k}))-\xi_{k-1}\right)^{2}\mbox{\bf 1}_{\left\{\varphi(X_{k})\geq \xi_{k-1}\right\}} \right. \\
 & & \left. \hspace*{6cm} - \left(\frac{1}{n} \sum_{k=1}^{n} \left(\Psi(\varphi(X_{k}))-\xi_{k-1}\right)\mbox{\bf 1}_{\left\{\varphi(X_{k})\geq \xi_{k-1}\right\}}\right)^{2}\right)
\end{eqnarray*}

\noindent where $(\xi_{n})_{n \geq 0}$ is the first component of~\eqref{RMVaR}. If $(A2)_{a}$ is satisfied for some $a\geq2$, then 
$$
\sigma_{n}^{2}\stackrel{a.s.}\longrightarrow \frac{1}{(1-\alpha)^{2}}\textnormal{ Var }\left(\left(\Psi(\varphi(X))-\xi^{*}_{\alpha}\right) \mbox{\bf{1}}_{\varphi(X)\geq \xi^{*}_{\alpha}}\right)
$$

\noindent and
\begin{eqnarray}
 \sqrt{n}\mbox{ } \frac{C_{n}-C^{*}_{\alpha}}{\sigma_{n}} \stackrel{\mathcal{L}} \longrightarrow \mathcal{N}(0,1).
\label{CLTbis}
\end{eqnarray}
\end{prop} 

\begin{proof}[\bf Proof.] The proof follows from standard arguments already used in the proof of the $a.s.$ convergence of the sequence $(C_{n})_{n\geq1}$ defined by \eqref{RMCVaR}.
\end{proof}

In practice, the convergence of the algorithm will be chaotic. The bottleneck of this algorithm is that it is only updated on rare events since it tries to measure the tail distribution of $\varphi(X)$ : $\mathbb{P}(\varphi(X)>\mbox{VaR}_{\alpha})=1-\alpha \approx 0$. Another problem may be the simulation of $\varphi(X)$. In practice, we have to deal with large portfolios of complex derivative securities and options. Each evaluation may require a lot of computational efforts and takes a long time. So, for practical implementation it is necessary to combine the above procedure with variance reduction techniques to achieve accurate results at a reasonable cost. The most appropriate technique when dealing with rare events is IS.

\medskip
\subsection{Some background on IS using stochastic approximation algorithm}

The second tool we want to introduce in this paper is a recursive IS procedure which increases the probability of simulations for which $\varphi(X)$ exceeds $\xi$. Our goal is to combine it adaptively with our first naive algorithm. Assume that $X$ has an absolutely continuous distribution $\mathbb{P}_{X}(dx)=p(x) \lambda_{d}(dx)$ where $\lambda_{d}$ denotes the Lebesgue measure on $(\mathbb{R}^{d},\it \mathcal{B}or(\mathbb{R}^{d}))$. The main idea of importance sampling by translation applied to the computation of $$\mathbb{E}[F(X)],$$ where $F \in L^{2}(\mathbb{P}_{X})$ satisfies $\mathbb{P}(F(X) \neq 0)>0$, is to use the invariance of the Lebesgue measure by translation, for every $\theta \in \mathbb{R}^{d}$,
\begin{equation}
\mathbb{E}[F(X)]=\mathbb{E}\left[F(X+\theta)\frac{p(X+\theta)}{p(X)}\right],
\label{invartranslation}
\end{equation}

\noindent and among all these random vectors with the same expectation, we want to select the one with the lowest variance, \emph{$i.e.$} the one with lowest quadratic norm

\begin{equation}
Q(\theta):=\mathbb{E}\left[F^{2}(X+\theta)\frac{p^{2}(X+\theta)}{p^{2}(X)}\right] \leq +\infty, \hspace*{0.4cm} \theta \in \mathbb{R}^{d}.
\label{nablaQ1v}
\end{equation}

\noindent If the following assumption

\begin{hyp2}
\forall \theta \in \mathbb{R}^{d}, \hspace*{0.8cm}  \mathbb{E}\left[ F^{2}(X) \frac{p(X)}{p(X-\theta)}\right]< +\infty
\label{B1}
\end{hyp2}

\noindent holds true, then $Q$ is everywhere finite and a reverse change of variable shows that:
\begin{eqnarray}
Q(\theta)=\mathbb{E}\left[ F^{2}(X) \frac{p(X)}{p(X-\theta)}\right],  \hspace*{0.4cm} \theta \in \mathbb{R}^{d}.
\label{ExpreVH1}
\end{eqnarray}

\noindent Now if $p$ satisfies 
\begin{hyp2}
                        \left\{
                        \begin{array}{l}
                            (i) \hspace*{0.3cm} \forall x \in \mathbb{R}^{d}, \mbox{  } \theta \mapsto p(x-\theta) \mbox{ is $\log$-\textsl{concave} }\\\\
                            (ii) \hspace*{0.2cm} \forall x \in \mathbb{R}^{d}, \mbox{ } \lim_{|\theta| \rightarrow +\infty } p(x-\theta)=0 \hspace*{0.3cm} \mbox { or }  \forall x \in \mathbb{R}^{d}, \mbox{ } \lim_{|\theta| \rightarrow +\infty} \frac{p(x-\theta)}{p^{2}(x-\frac{\theta}{2})}=0,
                        \end{array}
                        \right.
\label{B2}
\end{hyp2}

\noindent one shows that $Q$ is (strictly) finite, convex, goes to infinity at infinity so that $\arg \min Q={\left\{\nabla Q =0\right\}}$ is non empty (see \cite{arouna} and \cite{lemairepages}). Provided that $\nabla Q$ admits a representation as an expectation, then it is possible to devise a recursive RM procedure to approximate the optimal parameter $\theta^{*}$. Recursive IS by stochastic approximation has been first investigated by Kushner and then by several authors, see e.g. \cite{dufresne} and \cite{fu} in order to ``optimize'' or ``improve'' the change of measure in IS using a stochastic gradient RM algorithm based on the representation of $\nabla Q(\theta)$.

\noindent Recently, it has been brought back to light by Arouna (see \cite{arouna}) in the Gaussian case, based on the natural representation of $\nabla Q$ obtained by formally differentiating~\eqref{ExpreVH1}. Since we have no knowledge about the regularity of $F$ and do not wish to have any, we differentiate the second representation of $Q$ in~\eqref{ExpreVH1} and not~\eqref{nablaQ1v}. We obtain $\nabla Q(\theta)=\mathbb{E}\left[K(\theta,X)\right]$. 

 When $X=\mathcal{N}(0,1)$, $Q(\theta)=e^{\frac{|\theta|^{2}}{2}}\mathbb{E}[F^{2}(X)e^{-
\theta X}]$ so that $K(\theta,x)=e^{\frac{|\theta|^{2}}{2}}F^{2}(x)e^{-\theta x}(\theta-x)$. However, given this resulting form of $K$, the classical convergence results do not apply since
$||K(\theta,X)||_{2}$ is not sub-linear in $\theta$ (see condition~\eqref{LinearGrowthAssump} of Theorem 2.2). This induces the explosion of the procedure at almost
every implementation as pointed out in \cite{arouna}. This leads the author to introduce a ``constrained'' variant of the regular procedure based on repeated reinitializations known as the projection ``\`{a} la Chen''. It forces the stability of the algorithm and prevents explosion. Let us also mention a first alternative approach investigated in \cite{arouna} and \cite{bardou}, where Arouna and Bardou change the function to be minimized by introducing an entropy based criterion. Although it is only an approximation, it turns out to be often close to the original method. 

 Recently, Lemaire and Pag\`{e}s in \cite{lemairepages} revisited the original
approach and provided a new representation of $\nabla Q(\theta)$ for which the resulting $K(\theta,X)$ has a linear growth in $\theta$ so that all assumptions of Theorem 2.2 are satisfied. Thanks to \textsl{a third translation} of the variable
$\theta$, it is possible to plug back the parameter $\theta$ ``into'' $F$, the function $F$ having in common applications a known behaviour at infinity which makes possible to devise a ``regular'' and ``unconstrained'' stochastic algorithm. We will rely partially on this approach to devise our final procedure to compute both VaR and CVaR. To be more specific about the methodology proposed in \cite{lemairepages}, we introduce the following assumption on the probability density $p$ of $X$
\begin{hyp2}
\exists b \in [1,2] \mbox{ such that }
\left\{
\begin{array}{l}
     (i) \hspace*{0.3cm}  \frac{|\nabla p(x)|} {p(x)}  = O(|x|^{b-1})\hspace*{0.2cm} \mbox{ as } \hspace*{0.2cm} |x|\rightarrow \infty \\ \\
     (ii) \hspace*{0.3cm} \exists \rho>0, \log\left(p(x)\right) + \rho |x|^{b}\hspace*{0.2cm} \mbox{is convex,}
\end{array}
\right.
\label{B3}
\end{hyp2}

\noindent and introduce the assumption on $F$ :
\begin{hyp2}
\forall A >0 , \mathbb{E}\left[F(X)^{2} e^{A|X|^{b-1}}\right]< + \infty.
\end{hyp2}

\noindent One shows that as soon as (B1), (B2), (B3) and (B4) are satisfied, $Q_{1}$ and $Q_{2}$ are both finite and differentiable on $\mathbb{R}^{d}$ with a gradient given by
\begin{eqnarray}
\nabla Q(\theta)&:= & \mathbb{E}\left[F(X-\theta)^{2} \underbrace{\frac{p^{2}(X-\theta)}{p(X)p(X-2\theta)}\frac{\nabla p(X-2\theta)}{p(X-2\theta)}}_{W(\theta,X)}\right].
\label{nablaQ}
\end{eqnarray}

\noindent This expression may look complicated at first glance but in fact the weight term $W(\theta,X)$ can be easily controlled by a deterministic function of $\theta$ since 
\begin{equation}
|W(\theta,X)| \leq e^{2\rho|\theta|^{b}}(A|x|^{b-1}+A|\theta|^{b-1}+B)
\label{majorationdensity}
\end{equation}

\noindent for some real constants $A$ and $B$. In the case of a normal distribution $X\stackrel{d}{=}\mathcal{N}(0;1)$, 

$$
W(\theta,X)=e^{\theta^{2}}(2\theta-X).
$$ 

\noindent So, if we have a control on the growth of the function $F$, typically for some positive constant $c$
\begin{hyp2}
\left\{
\begin{array}{l}
  \forall x \in \mathbb{R}^{d}, |F(x)| \leq  G(x) \hspace*{0.2cm} \mbox{ and }   \hspace*{0.2cm} G(x+y) \leq C(1+G(x))^{c} (1+G(y))^{c} \\
   \mbox{ } \\
  \hspace*{4cm}  \mathbb{E}\left[|X|^{2(b-1)}G(X)^{4c}\right] < + \infty,\\
\end{array}
\right.
\label{B5c}
\end{hyp2}

\noindent then by setting 
\begin{equation}
\widetilde{W}(\theta,X):=\frac{e^{-2\rho|\theta|^{b}}}{1+G(-\theta)^{2c}} W(\theta,X),
\label{correctionpoids}
\end{equation}

\noindent we can define $K$ by 
\begin{equation}
K(\theta,x):=F(x-\theta)^{2}\widetilde{W}(\theta,X)
\end{equation}

\noindent so that it satisfies the linear growth assumption~\eqref{LinearGrowthAssump} of Theorem 2.2 and 
\begin{equation*}
\left\{{\theta \in \mathbb{R}^{d} \ | \ \mathbb{E}\left[K(\theta,X)\right]=0}\right\}=\left\{{\theta \in \mathbb{R}^{d}\ | \ \nabla Q(\theta)=0}\right\}.
\end{equation*}

\noindent Moreover, since $Q$ is convex $\nabla Q$ satisfies~\eqref{MeanRev}. Now we are in position to derive a recursive unconstrained RM algorithm 
\begin{equation}
\theta_{n}=\theta_{n-1}-\gamma_{n}K(\theta_{n-1},X_{n}), \ \theta_{0} \in \mathbb{R}^{d},
\label{URIS}
\end{equation}

\noindent that $a.s.$ converges to an $\arg\min Q$-valued (square integrable) random variable $\theta^{*}$.
\end{section}

\begin{section}{Design of a faster procedure: importance sampling and moving confidence level }

\subsection{Unconstrained adaptive importance sampling device}
We noted previously that the bottleneck in using the above algorithm lies in its very slow and chaotic convergence owing to the fact that $\mathbb{P}(\varphi(X)>\xi^{*}_{\alpha})= 1-\alpha$ is close to 0. This means that we observe fewer and fewer simulations for which $\varphi(X_{k})>\xi_{k-1}$ as the algorithm evolves. Thus, it becomes more and more difficult to compute efficiently some estimates of VaR$_{\alpha}$ and CVaR$_{\alpha}$ when $ \alpha \approx 1$. Moreover, in the bank and energy sectors, practitioners usually deal with huge portfolio made of hundreds or thousands of risk factors and options. The evaluation step of $\varphi(X)$ may be extremely time consuming. Consequently, to achieve accurate estimates of both VaR$_{\alpha}$ and CVaR$_{\alpha}$ with reasonable computational effort, the above algorithm~\eqref{VaRCVaRequalsteps} drastically needs to be speeded up by an IS procedure to ``recenter" the simulations where ``things do happen'', $i.e.$ which generates scenarios for which $\varphi(X)$ exceeds $\xi$.

 In this section we will focus on IS by mean translation. Our aim is to combine adaptively the IS (unconstrained) recursive procedure investigated in \cite{lemairepages} with our first ``naive'' approach described in~\eqref{VaRCVaRequalsteps}. Doing so every new sample is used to both optimize the IS change of measure and update VaR and CVaR procedures. We plan to minimize the asymptotic variance of both components of the algorithm (in its ``averaged'' form, as detailed in Theorem 2.4), namely 
\begin{equation}
\frac{\alpha(1-\alpha)}{f_{\varphi(X)}(\xi^{*}_{\alpha})}=\frac{\text{Var}(\textbf{1}_{\varphi(X)\geq \xi^{*}_{\alpha}})}{f_{\varphi(X)}(\xi^{*}_{\alpha})} \ \ \mbox{ for the } \textnormal{VaR}_{\alpha},
\end{equation} 

\noindent and,
\begin{equation}
\frac{\text{Var}((\Psi(\varphi(X))-\xi^{*}_{\alpha})\bf{1}_{\varphi(X)\geq \xi^{*}_{\alpha}})}{(1-\alpha)^{2}} \ \ \mbox{ for the } \textnormal{CVaR}_{\alpha},
\end{equation} 

\noindent provided the non-degeneracy assumption

\begin{hyp4} \forall \xi \in \arg \min V, \   \mathbb{P}\left(\left(\Psi(\varphi(X))-\xi\right)^{2}\mbox{\bf
1}_{\left\{\varphi(X) \geq \xi\right\}} >0\right)>0,
\end{hyp4}

\noindent holds. Since the density $f_{\varphi(X)}(\xi^{*}_{\alpha})$ is an intrinsic constant (and comes in fact from the Jacobian matrix $Dh(\xi^{*}_{\alpha},C^{*}_{\alpha})$ of the mean function $h$ of the algorithm) we are led to apply the IS paradigm described in Section 2.3 to
$$
F_{1}^{*}(X)=\textbf{1}_{\varphi(X)\geq \xi^{*}_{\alpha}} \ \ \mbox{ and } \ \ F_{2}^{*}(X)=\left(\Psi(\varphi(X))-\xi^{*}_{\alpha}\right) \textbf{1}_{\left\{\varphi(X)\geq \xi^{*}_{\alpha}\right\}}. 
$$
\noindent Let us temporary forget that of course we do not know $\xi^{*}_{\alpha}$ at this stage. Those two functionals are related to the minimization of the two convex functions
\begin{eqnarray}
Q_{1}(\theta,\xi^{*}_{\alpha})&:= &\mathbb{E}\left[\mbox{ \bf 1}_{\left\{\varphi(X) \geq \xi^{*}_{\alpha}\right\}} \frac{p(X)}{p(X-\theta)}\right]\\
Q_{2}(\mu,\xi^{*}_{\alpha})&:=
&\mathbb{E}\left[\left(\Psi(\varphi(X))-\xi^{*}_{\alpha}\right)^{2}\mbox{\bf
1}_{\left\{\varphi(X) \geq \xi^{*}_{\alpha}\right\}} \frac{p(X)}{p(X-\mu)}\right].
\label{Q2def}
\end{eqnarray}

\noindent We can apply  to these functions the minimizing procedure~\eqref{URIS} described at section 2.3. Since 
\begin{equation}
H_{1}\left(\xi^{*}_{\alpha},x\right)=1-\frac{1}{1-\alpha}F_{1}^{*}(x) \ \ \mbox{ and } \ \ H_{2}\left(\xi^{*}_{\alpha},C^{*}_{\alpha},x\right)=C^{*}_{\alpha}-\xi^{*}_{\alpha}-\frac{1}{1-\alpha}F_{2}^{*}(x)
\end{equation}

\noindent it is clear, owing to~\eqref{invartranslation} that 
$$
\mathbb{E}\left[H_{i}(\xi^{*}_{\alpha},X)\right]=\mathbb{E}\left[H_{i}\left(\xi^{*}_{\alpha},X+\theta\right)\frac{p(X+\theta)}{p(X)}\right] \ \ i=\ 1,\ 2.
$$

\noindent Now, since we do not know either $\xi^{*}_{\alpha}$ and $C^{*}_{\alpha}$ (the VaR$_{\alpha}$ and the CVaR$_{\alpha}$) respectively we make the whole procedure adaptive by replacing at step $n$, these unknown parameters by their running approximation at step $n-1$. This finally justifies to introduce the following global procedure. One defines the state variable, for $n\geq0$,
$$
Z_{n}:=\left(\xi_{n},C_{n},\theta_{n},\mu_{n}\right),
$$

\noindent where $\xi_{n},\ C_{n}$ denotes the VaR$_{\alpha}$ and the CVaR$_{\alpha}$ approximate, $\theta_{n}, \ \mu_{n}$ denotes the variance reducers for the VaR and the CVaR procedures. We update this state variable recursively by
\begin{equation}
Z_{n}=Z_{n-1}-\gamma_{n}L\left(Z_{n-1},X_{n}\right),
\label{algoglob}
\end{equation}

\noindent where $\left(X_{n}\right)_{n\geq1}$ is an i.i.d. sequence with distributions $X$ (and probability density $p$) and 
\begin{eqnarray}
L_{1}(\xi,\theta,x) & := & e^{-\rho|\theta|^{b}}\left(1-\frac{1}{1-\alpha} \mbox{ \bf 1}_{\left\{\varphi(x+\theta)\geq
\xi\right\}}\frac{p(x+\theta)}{p(x)}\right) ,\nonumber \\
L_{2}(\xi,C,\mu,x)  & := & C-\xi-\frac{1}{1-\alpha}(\Psi(\varphi(x+\mu))-\xi) \mbox{ \bf 1}_{\left\{\varphi(x+\mu) \geq \xi \right\}}\frac{p(x+\mu)}{p(x)}, \nonumber \\
L_{3}\left(\xi,\theta,x\right) & := & e^{-2 \rho |\theta|^{b}}\mbox{\bf 1}_{\left\{\varphi(x-\theta) \geq \xi\right\}} \frac{p^{2}(x-\theta)}{p(x)p(x-2\theta)}\frac{\nabla p(x-2\theta)}{p(x-2\theta)} ,  \\
L_{4}\left(\xi,\mu,x\right) & := & \frac{e^{-2\rho | \mu |^{b}}}{1+G(-\mu)^{2c}+\xi^{2}} \left(\Psi(\varphi(x-\mu))-\xi\right)^{2}  \nonumber \\
 &  & \hspace*{4cm} \times \mbox{\bf 1}_{\left\{\varphi(x-\mu) \geq \xi \right\}}\frac{p^{2}(x-\mu)}{p(x)p(x-2\mu)}\frac{\nabla p(x-2\mu)}{p(x-2\mu)}.
\label{HIS}
\end{eqnarray} 

\noindent The following proposition establishes the $a.s.$ convergence of the procedure. For the sake of simplicity we will assume the uniqueness of the VaR$_{\alpha}$ of $\varphi(X)$.

\begin{prop}{(Efficient computation of VaR and CVaR).}
Suppose that $\Psi(\varphi(X)) \in L^{2}\left(\mathbb{P}\right)$, that the distribution function of $\varphi(X)$ is continuous and increasing (so that $\textnormal{VaR}_{\alpha}(\varphi(X))$ is unique) and that (A3) holds. Assume that, for every $\xi \in \mathbb{R}$, $Q_{i}(.,\xi)$ (i=1,2) satisfies (B1), $i.e.$
\begin{equation}
\forall \theta \in \mathbb{R}^{d}, \ \ \mathbb{E}\left[\left(1+\left(\Psi(\varphi(X))-\xi\right)^{2}\right)\textbf{1}_{\left\{\varphi(X)\geq\xi \right\}}\frac{p(X)}{p(X-\theta)}\right] < +\infty.
\end{equation}

\noindent Suppose that $p$ satisfies (B2) and (B3) and that 
$$
\forall A>0, \mathbb{E}\left[\left(\Psi(\varphi(X))^2+1\right)e^{A|X|^{b-1}}\right]<+\infty.
$$

\noindent Assume that the step sequence $(\gamma_{n})_{n\geq1}$ satisfies (A1).
\noindent Then, $$Z_{n}\stackrel{a.s.}{\longrightarrow}z^{*}:=(\xi^{*}_{\alpha},C^{*}_{\alpha},\theta^{*}_{\alpha},\mu^{*}_{\alpha})$$
\noindent where $\xi^{*}_{\alpha}=\textnormal{VaR}_{\alpha}(\varphi(X))$, $C^{*}_{\alpha}=\Psi\textnormal{-CVaR}_{\alpha}(\varphi(X))$ and $(\theta^{*}_{\alpha},\mu^{*}_{\alpha})$ are the optimal variance reducers (to be precise some random vectors taking values in $\left\{{\nabla Q_{1}(\xi^{*}_{\alpha},.)=0}\right\}$ and $\left\{{\nabla Q_{2}(\xi^{*}_{\alpha},.)=0}\right\}$ respectively). 
\end{prop}

\begin{proof}[\bf Proof.] We first prove the $a.s.$ convergence of the 3-tuple $(\xi_{n}, \theta_{n}, \mu_{n})$ that of $(C_{n})_{n\geq1}$ will follow by the same arguments used in the proof in Section 2.2. The mean function $l$ is defined by 
$$
l(\xi,\theta,\mu):=\left(e^{-\rho|\theta|^{b}}\left(1-\frac{1}{1-\alpha}\mathbb{P}\left(\varphi(X)\geq\xi\right)\right),e^{-2\rho|\theta|^{b}}\nabla Q_{1}\left(\theta,\xi\right), \frac{e^{-2\rho|\mu|^{b}}}{1+G(-\mu)^{c}+\Psi(\xi)^{2}}\nabla Q_{2}(\mu,\xi)\right),
$$
\noindent hence,
$$
\mathcal{T}^{*}=\left\{{l=0}\right\}=\left\{{\xi^{*}_{\alpha}}\right\}\times\left\{{\nabla Q_{1}(\xi^{*}_{\alpha},.)=0}\right\}\times\left\{{\nabla Q_{2}(\xi^{*}_{\alpha},.)=0}\right\}.
$$ 

\noindent In order to apply the extended Robbins-Monro Theorem, we have to check the following facts:

\medskip
\noindent $\bullet$  \textsl{Mean reversion}: One checks that $\forall \zeta=(\xi,\theta,\mu) \in \mathbb{R}\times\mathbb{R}^{d}\times\mathbb{R}^{d} \mbox{ \textbackslash }\mathcal{T}^{*}$, $\forall \zeta^{*} \in \mathcal{T}^{*}$,

\begin{eqnarray*}
\left\langle \zeta-\zeta^{*},l(\zeta)\right\rangle & = & e^{-\rho|\theta|^{b}}(\xi-\xi^{*}_{\alpha})\frac{(\mathbb{P}(\varphi(X)\leq\xi)-\alpha)}{1-\alpha}+\frac{e^{-2\rho |\theta|^{b}}}{1-\alpha}\left\langle \theta-\theta^{*}_{\alpha},\nabla Q_{1}(\theta,\xi)\right\rangle \\
& & +\mbox{ } \frac{e^{-2\rho |\mu|^{b}}}{(1-\alpha)(1+F(-\mu)^{2c})}\left\langle \mu-\mu^{*}_{\alpha},\nabla Q_{2}(\mu,\xi)\right\rangle>0, 
\end{eqnarray*}

\noindent owing to the convexity of $\theta \mapsto Q_{1}(\theta,\xi)$ and $\mu \mapsto Q_{2}(\mu,\xi)$, for every $\xi\in\mathbb{R}$.

\medskip
\noindent $\bullet$ \textsl{Linear growth:} Let us first deal with $L_{1}$. First note that:
$$\mathbb{E}\left[L_{1}\left(\xi,\theta,X\right)^{2}\right]\leq C\left(1+\mathbb{E}\left[e^{-2\rho|\theta|^{b}}\mbox{\bf 1}_{\left\{\varphi(X+\theta)\geq\xi\right\}}\frac{p^{2}(X+\theta)}{p^{2}(X)}\right]\right)\leq C\left(1+\mathbb{E}\left[e^{-2\rho|\theta|^{b}}\frac{p(X)}{p(X-\theta)}\right]\right).$$ 

\noindent Now, elementary computations show (see \cite{lemairepages} for more details) that (B3)(ii) implies that 
$$ 
\frac{p^{2}(x)}{p(x-\theta)} \leq 	e^{2\rho|\theta|^{b}}p(x+\theta),
$$

\noindent so that 
$$
\mathbb{E}\left[e^{-2\rho|\theta|^{b}}\frac{p(X)}{p(X-\theta)}\right]\leq \mathbb{E}\left[\frac{p(X+\theta)}{p(X)}\right]=1.
$$

\noindent $L_{3}$ and $L_{4}$ can be treated by a straightforward adaptation of the proofs in \cite{lemairepages}. Then, one can apply Theorem 2.2 which yields the announced result for $\left(\xi_{n},\theta_{n},\mu_{n}\right)$.
\noindent The $a.s.$ convergence of $C_{n}$ toward $C^{*}_{\alpha}$ can be deduced from the $a.s.$ convergence of the series
$$
M_{n}^{\gamma}:=\sum_{k=1}^{n} \gamma_{k} \Delta \widetilde{M}_{k}, \mbox { } n \geq 1,
$$

\noindent where $\Delta \widetilde{M}_{n}$ are martingale increments defined by
\begin{eqnarray*}
\Delta \widetilde{M}_{n} & = & \mathbb{E}[(\Psi(\varphi(X))-\xi)\mbox{ \bf 1}_{\left\{\varphi(X) \geq \xi\right\}}]_{|\xi=\xi_{n-1}}  \\
 & & \hspace*{1cm} -(\Psi(\varphi(X_{n}+\mu_{n-1}))-\xi_{n-1})\mbox{ \bf 1}_{\left\{\varphi(X_{n}+\mu_{n-1})\geq \xi_{n-1}\right\}}\frac{p(X_{n}+\mu_{n-1})}{p(X_{n})}, \ \ n\geq1,
\end{eqnarray*}

\noindent satisfying 
$$
\mathbb{E}\left[\Delta \widetilde{M}_{n}^{2} | \mathcal{F}_{n-1}\right] \leq  \mathbb{E}\left[(\Psi(\varphi(X+\mu))-\xi)\mbox{ \bf 1}_{\left\{\varphi(X+\mu)\geq \xi\right\}}\frac{p(X+\mu)}{p(X)}\right]_{|\xi=\xi_{n-1}, \theta=\theta_{n-1}, \mu=\mu_{n-1}}.
$$
\noindent We conclude by the same arguments used in the proof in Section 2.2. 
\end{proof}

 Now, we are interested by the rate of convergence of the procedure. It shows that the algorithm behaves as expected under quite standard assumptions: it satisfies a Gaussian CLT with optimal rate and minimal variances.


\begin{thm} Suppose the assumptions of Proposition 3.1 hold true. Assume that $\Psi(\varphi(X)) \in L^{2a}(\mathbb{P})$ for some $a>1$ and that the step sequence is $\gamma_{n}=\frac{\gamma_{1}}{n^{p}}$ with $ \frac{1}{2} < p <1$ and $\gamma_{1}>0$. Suppose that the density $f_{\varphi(X)}$ is continuous and strictly positive on its support. Let $(\overline{\xi}_{n},\overline{C}_{n})_{n\geq1}$ be the sequence of Cesaro means defined by:
$$
\overline{\xi}_{n}:=\frac{\xi_{0}+...+\xi_{n-1}}{n}, \ \ \ \overline{C}_{n}:=\frac{C_{0}+...+C_{n-1}}{n}, \ \ n\geq1.
$$

\noindent This sequence satisfies the following CLT:
\begin{equation}{}
                \sqrt{n} 
            \left(
            \begin{array}{c}
                \overline{\xi}_{n}-\xi^{*}_{\alpha} \\
                \overline{C}_{n}-C^{*}_{\alpha}
            \end{array}
            \right)\stackrel{\mathcal{L}}{\rightarrow} \mathcal{N}(0,\Sigma^{*}) \hspace*{.5cm} \mbox{ as } n \rightarrow +\infty,
\label{TCLopt}
\end{equation}

\noindent where 
\begin{eqnarray*}
\Sigma^{*}_{1,1} & = & \frac{1}{f_{\varphi(X)}^{2}(\xi^{*}_{\alpha})} \textnormal{Var}\left(\mbox{ \bf 1}_{\left\{\varphi(X+\theta^{*}_{\alpha})\geq \xi^{*}_{\alpha}\right\}}\frac{p(X+\theta^{*}_{\alpha})}{p(X)}\right), \vspace*{0.2cm}\\
\Sigma^{*}_{1,2} & = &  \Sigma^{*}_{2,1} =\frac{1}{(1-\alpha)f_{\varphi(X)}(\xi^{*}_{\alpha})}\textnormal{Cov}\left( \left(\Psi(\varphi(X+\mu^{*}_{\alpha}))-\xi^{*}_{\alpha}\right)\mbox{\bf 1}_{\left\{\varphi(X+\mu^{*}_{\alpha})>\xi^{*}_{\alpha}\right\}}\frac{p(X+\mu^{*}_{\alpha})}{p(X)}, \right. \\
& & \hspace*{5cm} \left. \mbox{\bf 1}_{\left\{\varphi(X+\theta^{*}_{\alpha})\geq \xi^{*}_{\alpha}\right\}}\frac{p(X+\theta^{*}_{\alpha})}{p(X)}\right), \vspace*{0.2cm} \\
\Sigma^{*}_{2,2} & = & \frac{1}{(1-\alpha)^{2}}\textnormal{Var}\left(\left(\Psi(\varphi(X+\mu^{*}_{\alpha}))-\xi^{*}_{\alpha}\right) \mbox{ \bf 1}_{\left\{\varphi(X+\mu^{*}_{\alpha}) \geq \xi^{*}_{\alpha}\right\}}\frac{p(X+\mu^{*}_{\alpha})}{p(X)}\right).  
\end{eqnarray*}
\end{thm}

\begin{proof}[\textbf{ Proof.}] The proof is built like the one of Theorem 2.4. If we denote $h$ the mean function of the global algorithm $h(z)=\mathbb{E}[L(z,X)]$, the algorithm~\eqref{algoglob} can be written as 
\begin{eqnarray}
Z_{n}=Z_{n-1}-\gamma_{n}\left(h(Z_{n-1})+\tilde{\epsilon}_{n}\right), \  n\geq 1, \ \ 
Z_{0}=\left(\xi_{0},0\right),  \  \xi_{0} \in L^{1}(\mathbb{P}),
\label{RMstd}
\end{eqnarray} 

\noindent where the first two components of $h$ are the same function as the ones in the proof of Theorem 2.4 and $(\tilde{\epsilon}_{n})_{n\geq1}$ denotes the $\mathcal{F}_{n}$-adapted martingale increments sequence where
\begin{eqnarray*}
\tilde{\epsilon}_{1,n} & := & \frac{1}{1-\alpha}\left(\mathbb{P}\left(\varphi(X)\geq\xi\right)_{| \xi=\xi_{n}}-\mbox{ \bf 1}_{ \left\{\varphi(X_{n+1}+\theta_{n}) \geq \xi_{n} \right\} }\frac{p(X_{n+1}+\theta_{n})}{p(X_{n+1})}\right) , \vspace*{0.2cm}\\
\tilde{\epsilon}_{2,n} & := & \frac{1}{1-\alpha}\left(\mathbb{E}[(\Psi(\varphi(X))-\xi)\mbox{ \bf 1}_{\left\{\varphi(X) \geq \xi\right\}}]_{|\xi=\xi_{n}} \right. \\ 
& & \hspace*{3cm} \left.  -(\Psi(\varphi(X_{n+1}+\mu_{n}))-\xi_{n})\mbox{ \bf 1}_{\left\{\varphi(X_{n+1}+\mu_{n})\geq \xi_{n}\right\}}\frac{p(X_{n+1}+\mu_{n})}{p(X_{n+1})}\right).
\end{eqnarray*}

\noindent One can check easily that the sequence $(\tilde{\epsilon}_{n})_{n\geq1}$ satisfies $(i)-(iv)$ of~\eqref{hypii}.
\end{proof}

\noindent \textbf{Remarks:} $\bullet$ There exists a CLT for the whole sequence $(Z_{n})_{n\geq1}$ and for its empirical mean $(\overline{Z}_{n})_{n\geq1}$ according to Ruppert and Polyak averaging principle. We only stated the result for the two components of interest (the ones which converge to VaR and CVaR respectively) since we only need rough estimates for the other two (see below).

\smallskip
\noindent $\bullet$ In the first Central Limit Theorem (Theorem 2.4) for
quantile estimation, the factor $\alpha(1-\alpha)$ is the variance
of the indicator function of the event $\left\{{\varphi(X) \geq \xi^{*}_{\alpha}}\right\}$. With our recursive IS procedure, it is replaced by the variance of the shifted indicator function modified by the measure change: $\mbox{Var}\left(\mbox{\bf 1}_{\left\{\varphi(X+\theta^{*}_{\alpha}) > \xi^{*}_{\alpha}\right\}}\frac{p(X+\theta^{*}_{\alpha})}{p(X)}\right).$ For further details on the rate of convergence of the unconstrained recursive importance sampling procedure, we refer to \cite{lemairepages}.

\bigskip
 Now, let us point out an important issue. The algorithm~\eqref{algoglob} raises an important problem numerically speaking. Actually, we have two algorithm $\xi_{n}$ and $(\theta_{n},\mu_{n})$ that are in competitive conditions, $i.e.$ on one hand, we added an IS procedure to $(\xi_{n})_{n\geq1}$ to improve the convergence toward $\xi^{*}_{\alpha}$, and on the other hand, the adjustment of the parameters $(\theta_{n},\mu_{n})$ ``need'' some samples $X_{n+1}$ satisfying $\varphi(X_{n+1}-\theta_{n})> \xi_{n}$ and $\varphi(X_{n+1}-\mu_{n})>\xi_{n}$ ($\Psi \equiv Id$) which tend to become rare events. Somehow, we postponed the problems resulting from rare events on the IS procedure itself which may ``freeze''. This in term suggests to break the link between the VaR-CVaR and the IS procedures by introducing a VaR \emph{companion procedure} that will drive the IS parameters to the tail distribution. A solution to do this is to make the confidence level increase slowly from a lower value (say $\alpha_{0}=50\%$) up to the target level $\alpha$. This kind of incremental threshold increase has been already proposed in \cite{kroese} in a different framework. This idea is developed in the next section.

\subsection{How to control the move towards the critical risk area: the final procedure}
From a theoretical point of view, so far, we considered the purely adaptive approach where we approximate $(\xi^{*}_{\alpha},C^{*}_{\alpha},\theta^{*}_{\alpha},\mu^{*}_{\alpha})$ using the \emph{same innovation sequences}. From a numerical point of view, we only need a rough estimate of the optimal IS parameters $(\theta^{*}_{\alpha},\mu^{*}_{\alpha})$. So that we are led to break the algorithm into two phases. Firstly, we compute a rough estimate of the optimal IS parameters $\left(\theta_{M},\mu_{M}\right)$ with a small number of iterations $M$ and in a second time, estimate the VaR$_{\alpha}$ and the CVaR$_{\alpha}$ with those optimized parameters with $N$ iterations ($M\ll N$ in practice).

Now, in order to circumvent the problem induced by the IS procedure, we propose to introduce companion VaR procedure (without IS, $i.e.$, based on $H_{1}$ from Section 2.2) that will lead the IS parameters into the critical risk area during a first phase of the simulation, say the first $M$ iterations. An idea to control the growth of $\theta_{n}$ and $\mu_{n}$ at the beginning of the algorithm, since we have no idea on how to twist the distribution of $\varphi(X)$, is to move slowly toward the target critical risk area (at level $\alpha$) in which $\varphi(X)$ exceeds $\xi$ by introducing a non-decreasing sequence $\alpha_{n}$ slowly converging to $\alpha$ during the first phase. Since the algorithm for the CVaR component $C_{n}$ is free of $\alpha$, by doing so, we only modify the VaR procedure $\xi_{n}$. The function $H_{1}$ in~\eqref{VaRCVaRequalsteps} is replaced by its counterpart which depends on the moving confidence level $\alpha_{n}$, namely 
\begin{equation}
\hat{\xi}_{n}= \hat{\xi}_{n-1}-\gamma_{n} \hat{H}_{1}\left(\hat{\xi}_{n-1},X_{n},\alpha_{n}\right),\ \ n\geq1, \hat{\xi}_{0}=\xi_{0}\in
L^{1}(\mathbb{P}). 
\label{RMVaRAlphaSequence}
\end{equation}

\noindent where,
$$
\forall \ \xi \in \mathbb{R}, \ \forall x \in \mathbb{R}^{d}, \forall \ \hat{\alpha} \in ]0,1[, \   \hat{H}_{1}\left(\xi,x,\hat{\alpha}\right)=1-\frac{1}{1-\hat{\alpha}} \mbox{ \bf
1}_{\left\{\varphi(x) \geq \xi\right\}}.
$$ 
 
\noindent The sequence $\left(\hat{\xi}_{n}\right)_{n\geq0}$ is only designed to drive ``smoothly'' the IS procedures toward the ``critical area'' at the beginning of the procedure, say during the first $M$ iterations and in no case to approximate $\xi^{*}_{\alpha}$ or $C^{*}_{\alpha}$. To be more precise, we define recursively the variance reducer sequence $(\hat{\theta}_{n})_{n\geq1}$, $(\hat{\mu}_{n})_{n\geq1}$ by plugging at each step $n$, $\hat{\xi}_{n-1}$ into $L_{3}(., \hat{\theta}_{n-1}, X_{n})$ and $L_{4}(., \hat{\mu}_{n-1}, X_{n})$ as defined in Section 3.1. This reads as follows, for $n\geq1$, 
\begin{equation}
\left\{
\begin{array}{l}
\hat{\xi}_{n}= \hat{\xi}_{n-1}-\gamma_{n} \hat{H}_{1}\left(\hat{\xi}_{n-1},X_{n},\alpha_{n}\right), \ \ \hat{\xi}_{0} \in L^{1}(\mathbb{P}), \vspace*{0.2cm}\\
\hat{\theta}_{n}=\hat{\theta}_{n-1}-\gamma_{n}L_{3}\left(\hat{\xi}_{n-1},\hat{\theta}_{n-1},X_{n}\right),  \ \ \theta_{0} \in \mathbb{R}^{d} \vspace*{0.2cm}, \\
\hat{\mu}_{n}=\hat{\mu}_{n-1}-\gamma_{n}L_{4}\left(\hat{\xi}_{n-1},\hat{\mu}_{n-1},X_{n}\right), \ \ \mu_{0} \in \mathbb{R}^{d}.
\end{array}
\label{UAISAlphaSequence}
\right.
\end{equation}

\noindent Although, we are not really interested in the asymptotic of this procedure $(\hat{\xi}_{n})$, its theoretical convergence follows from Theorem 2.2: as a matter of fact if we define a remainder term $r_{n}$ by:
$$
r_{n}:=\hat{H}_{1}\left(\hat{\xi}_{n-1},X_{n},\alpha_{n}\right)-H_{1}\left(\hat{\xi}_{n-1},X_{n}\right), \hspace*{0.2cm} n\geq1,
$$  
 
\noindent the procedure defined by~\eqref{UAISAlphaSequence} now reads
\begin{equation}{}
\hat{\xi}_{n}= \hat{\xi}_{n-1}-\gamma_{n}
(H_{1}(\hat{\xi}_{n-1},X_{n})+r_{n}), \ \ 
n\geq1, \ \ \hat{\xi}_{0}\in L^{1}(\mathbb{P}).
\label{RMVaRAlphaSequence2}
\end{equation}

\noindent One checks that 
$$
|r_{n}|\leq \frac{\left|\alpha_{n}-\alpha\right|}{(1-\alpha)^{2}},
$$

\noindent so that Assumption~\eqref{condremaindersequ} of Theorem 2.2 is satisfied as soon as
$$
\sum_{n\geq1} \gamma_{n} (\alpha-\alpha_{n})^{2}<+\infty.
$$

\subsection{A final procedure for practical implementation}

In practice, we divided our procedure into two phases: 

\noindent $\rhd$ Phase I is devoted to the estimation of the variance reducers $(\theta^{*}_{\alpha},\mu^{*}_{\alpha})$ using~\eqref{UAISAlphaSequence}. The moving confidence level $\alpha$ has been settled as follows ($M\approx15000$) :
$$
\alpha_{n}=50\% \mbox{ for } \ 1\leq n \leq M_{1}:=M/3, \ \alpha_{n}=80\% \mbox{ for } M_{1}<n\leq 2M_{1}, \ \alpha_{n}=\alpha \mbox{ for }  \ 2M_{1}<n\leq M.
$$  

\noindent $\rhd$ Phase II produces some estimates for $(\xi^{*}_{\alpha},C^{*}_{\alpha})$ based on the procedure defined by~\eqref{algoglob} and its Cesaro mean with $N$ iterations. Note that during this phase, we keep on updating the IS parameters adaptively.

Now, we can summarize the two phase of the final procedure by the following pseudo-code:

\begin{algorithm*}
\begin{algorithmic}
\STATE{Phase I: Estimation of $(\mu^{*}_{\alpha},\theta^{*}_{\alpha})$. $M\ll N$ (typically $M\approx N/100$).}
\FOR{$n=1$ to $M$} 
\STATE{$ \hat{\xi}_{n}  =  \hat{\xi}_{n-1}-\gamma_{n} \hat{H}_{1}\left(\hat{\xi}_{n-1},X_{n},\alpha_{n}\right), $}
\STATE{$\hat{\theta}_{n}  =  \hat{\theta}_{n-1}-\gamma_{n}L_{3}\left(\hat{\xi}_{n-1},\hat{\theta}_{n-1},X_{n}\right), $} 
\STATE{$\hat{\mu}_{n}  = \hat{\mu}_{n-1}-\gamma_{n}L_{4}\left(\hat{\xi}_{n-1},\hat{\mu}_{n-1},X_{n}\right). $}
\ENDFOR

\medskip

\STATE{Phase II: Estimation of $(\xi^{*}_{\alpha},C^{*}_{\alpha})$. Set, for instance, $\xi_{0}=\hat{\xi}_{M}$, $C_{0}=0$, $\theta_{0}=\hat{\theta}_{M}$, and $\mu_{0}=\hat{\mu}_{M}$.}
\FOR{$n=1$ to $N$}  
\STATE{$ \xi_{n}        =  \xi_{n-1}       - \gamma_{n}   L_{1}\left(\xi_{n-1},\theta_{n-1},X_{n}\right), $}
\STATE{$ C_{n}          =  C_{n-1}         - \gamma_{n}   L_{2}\left(\xi_{n-1},C_{n-1},\mu_{n-1},X_{n}\right),$}
\STATE{$ \theta_{n}     = \theta_{n-1}    - \gamma_{n}   L_{3}\left(\xi_{n-1},\theta_{n-1}, X_{n}\right), $}
\STATE{$ \mu_{n}        =  \mu_{n-1}       - \gamma_{n}   L_{4}\left(\xi_{n-1},\mu_{n-1}, X_{n}\right), $}

\medskip

\STATE{Compute the Cesaro means}

\medskip

\STATE{$  \bar{\xi}_{n}  =  \bar{\xi}_{n-1} - \frac{1}{n} \left(\bar{\xi}_{n-1}-\xi_{n}\right), $}
\STATE{$  \bar{C}_{n}    =  \bar{C}_{n-1}   - \frac{1}{n}  \left(\bar{C}_{n-1}-C_{n}\right). $}
\ENDFOR

\medskip

\STATE{$(\xi^{*}_{\alpha},C^{*}_{\alpha})$ is    estimated by $(\bar{\xi}_{N},\bar{C}_{N})$.}

\end{algorithmic}
\end{algorithm*}

An alternative, especially as concerns practical implementation, is to replace to Phase~II by 

\medskip
\noindent {\bf Phase~II' in which  the variance reducers coming from Phase~I  are frozen at $\hat{\theta}_{M}$, $\hat{\mu}_{M}$}. The only updated sequence is ($\xi_{n},C_{n}$),  as follows
\begin{eqnarray*}
\xi_{n}       & = & \xi_{n-1}       - \gamma_{n}   L_{1}\left(\xi_{n-1},\hat{\theta}_{M},X_{n}\right), \vspace*{.15cm} \\
C_{n}         & = & C_{n-1}         - \gamma_{n}   L_{2}\left(\xi_{n-1},C_{n-1},\hat{\mu}_{M},X_{n}\right). 
\end{eqnarray*}
\end{section}

%

\begin{section}{Towards some extensions}

\subsection{Extension to exponential change of measure: the Esscher transform}
Considering an exponential change of measure (also called Esscher transform) instead of the mean translation is a rather natural idea that has already been investigated in \cite{kawai} and \cite{lemairepages} to extend the constrained IS stochastic approximation algorithm with repeated projections introduced in \cite{arouna}. We briefly introduce the framework and give the main results without any proofs (for more details, see \cite{lemairepages} and \cite{frikha}). Let $\psi$ denote the cumulant generating function (or $\log$-Laplace) of $X$ \textsl{$i.e.$} the function defined by $\psi(\theta):=\log\mathbb{E}[e^{\left\langle \theta,X\right\rangle}]$. We assume that $\psi(\theta)< +\infty$, which implies that $\psi$ is an infinitely differentiable convex function and define 
$$
p_{\theta}(x)=e^{\left\langle \theta,x\right\rangle-\psi(\theta)}p(x), \hspace*{0.3cm} x\in \mathbb{R}^{d}.
$$

\noindent We denote by $X^{(\theta)}$ any random variable with distribution $p_{\theta}$. We make the following assumption on the function $\psi$ 
\begin{hyppsi}
\lim_{|\theta|} \psi(\theta)-2\psi\left(\frac{\theta}{2}\right)=+\infty \hspace*{0.3cm} \mbox{ and } \hspace*{0.3cm} \exists \delta>0,\mbox{ } \theta \mapsto \psi(\theta)-\delta |\theta|^{2} \mbox{ is concave.}
\end{hyppsi}

\noindent The two functionals to be minimized are 
\begin{eqnarray}
Q_{1}(\theta,\xi^{*}_{\alpha}) & := & \mathbb{E}\left[\mbox{\bf 1}_{\left\{\varphi(X)>\xi^{*}_{\alpha}\right\}}e^{-\left\langle \theta,X\right\rangle+\psi(\theta)}\right] \\
Q_{2}(\mu,\xi^{*}_{\alpha}) & := & \mathbb{E}\left[(\Psi(\varphi(X))-\xi^{*}_{\alpha})^{2}\mbox{\bf 1}_{\left\{\varphi(X)>\xi^{*}_{\alpha}\right\}}e^{-\left\langle \mu,X\right\rangle+\psi(\mu)}\right].
\end{eqnarray}

\noindent According to Proposition 3 in \cite{lemairepages} as soon as $\psi$ satisfies $(H^{es}_{\delta})$ and that, 
\begin{equation}
\forall \xi \in \mathbb{R}, \forall \theta \in \mathbb{R}^{d}, \hspace*{0.3cm}  \mathbb{E}[|X| \left(1+\Psi(\varphi(X))^{2}\right)e^{\left\langle \theta,X\right\rangle}]<+\infty,
\label{hypfunction}
\end{equation}

\noindent for every $\xi \in \mathbb{R}$, the functions $Q_{1}(.,\xi)$ and $Q_{2}(.,\xi)$ are finite, convex, differentiable on $\mathbb{R}^{d}$, go to infinity at infinity, so that $\arg \min Q_{1}(.,\xi)$ and $\arg \min Q_{2}(.,\xi)$ are non empty. Moreover, their gradients are given by 
\begin{eqnarray}
\nabla_{\theta} Q_{1}(\theta,\xi) & = & \mathbb{E}\left[(\nabla \psi(\theta)-X^{(-\theta)})\mbox{\bf 1}_{\left\{\varphi(X^{(-\theta)})>\xi\right\}}\right]e^{\psi(\theta)-\psi(-\theta)} \\
\nabla_{\mu} Q_{2}(\mu,\xi) & = & \mathbb{E}\left[(\nabla \psi(\mu)-X^{(-\mu)})(\Psi(\varphi(X^{(-\mu)}))-\xi)^2\mbox{\bf 1}_{\left\{\varphi(X^{(-\mu)})>\xi\right\}}\right]e^{\psi(\mu)-\psi(-\mu)}
\end{eqnarray}

\noindent with $\nabla \psi(\theta)=\frac{\mathbb{E}[Xe^{\left\langle \theta,X\right\rangle}]}{\mathbb{E}[e^{\left\langle \theta,X\right\rangle}]}.$ Now, the main result of this section is the following theorem (for more details, we refer to \cite{lemairepages} and \cite{frikha}).

\begin{thm}Suppose that $\psi$ satisfies $(H^{es}_{\delta})$ and that $(A2)_{1}$, (A3) hold. Assume that~\eqref{hypfunction} is fulfilled and that 
$$\forall x \in \mathbb{R}^{d}, |\Psi(\varphi(x))| \leq Ce^{\frac{\lambda}{4}|x|}\ \ \mbox{ and } \ \ \mathbb{E}[|X|^{2}e^{\lambda |X|}]<+\infty.$$

\noindent One considers   the recursive procedure  
\begin{equation}
Z_{n}=Z_{n-1}-\gamma_{n}L(Z_{n-1},X_{n}), \;n\ge1, \;  Z_{0}=(\xi_{0},C_{0},\theta_{0},\mu_{0})
\label{EsVaRCVaR}
\end{equation}

\noindent where $(\gamma_{n})_{n\geq1}$ satisfies the usual step assumption (A1), $Z_{n}:=(\xi_{n},C_{n},\theta_{n},\mu_{n})$ and each component of $L$ is defined by  
\begin{equation*}
\left\{\begin{array}{l}
L_{1}\left(\xi_{n-1},\theta_{n-1},X^{(\theta_{n-1})}_{n}\right):=e^{-\frac{\psi(\theta_{n-1})+\psi(-\theta_{n-1})}{2}}\left(1-\frac{1}{1-\alpha}\mbox{\bf 1}_{\left\{\varphi(X^{(\theta_{n-1})}_{n})>\xi_{n-1}\right\}}\mbox{ }e^{\psi(\theta_{n-1})-\left\langle X^{(\theta_{n-1})}_{n},\theta_{n-1}\right\rangle}\right), \vspace*{.15cm} \\
L_{2}\left(\xi_{n-1},C_{n-1},\mu_{n-1},X^{(\mu_{n-1})}_{n}\right):=C-\bar{w}(\xi_{n-1},\mu_{n-1},X^{(\mu_{n-1})}_{n}), \vspace*{.15cm} \\
L_{3}\left(\xi_{n-1},\theta_{n-1},X^{(-\theta_{n-1})}\right):=\mbox{\bf 1}_{\left\{\varphi(X^{(-\theta_{n-1})})>\xi_{n-1}\right\}}(\nabla \psi(\theta_{n-1})-X^{(-\theta_{n-1})}), \vspace*{.15cm} \\
L_{4}\left(\xi_{n-1},\mu_{n-1},X^{(-\mu_{n-1})}\right):=\frac{e^{-\frac{\lambda}{2}\sqrt{d}|\nabla \psi(-\mu_{n-1})|}}{1+\xi_{n-1}^2}(\Psi(\varphi(X^{(-\mu_{n-1})}))-\xi_{n-1})^{2}\mbox{\bf 1}_{\left\{\varphi(X^{(-\mu_{n-1})})>\xi_{n-1}\right\}}  \vspace*{.15cm}\\
		\hspace*{5cm} \times(\nabla \psi(\mu_{n-1})-X^{(-\mu_{n-1})}), 
\end{array}
\right.
\end{equation*}

\noindent  with $\displaystyle 
\bar{w}(\xi,\mu,x):=\Psi(\xi)+\frac{1}{1-\alpha}(\Psi(\varphi(x))-\Psi(\xi))\mbox{\bf 1}_{\{\varphi(x)>\xi\}}\mbox{ }e^{\psi(\mu)-\left\langle \mu,x\right\rangle}$.

\smallskip Then, $Z_n$  converges $a.s.$ toward $z^{*}:=\left(\xi^{*}_{\alpha},C^{*}_{\alpha},\theta^{*}_{\alpha},\mu^{*}_{\alpha}\right)$, where $\xi^{*}_{\alpha}$ is a square integrable $\textnormal{VaR}_{\alpha}$-valued random variable, $C^{*}_{\alpha}= \Psi\textnormal{-CVaR}_{\alpha}(\varphi(X))$, $\theta^{*}_{\alpha}$ is a (square integrable)  $\arg \min Q_{1}(.,\xi^{*}_{\alpha})$-valued random vector and $\mu^{*}_{\alpha}$ is a (square integrable)  $\arg \min Q_{2}(.,\xi^{*}_{\alpha})$-valued random vector.

\end{thm}
 
%

\subsection{Extension to infinite dimensional setting}
In the above sections, we proposed our algorithm in a finite dimensional setting where the value of the loss $L=\varphi(X)$ is a function of a random vector having values in $\mathbb{R}^{d}$. This is due to the fact that generally the value of a portfolio may depend on a finite number of decisions taken in the past. Thus, the value of the loss at the horizon time $T-t$ may depend on a large number of dates in the past $t_{0}=t<t_{1}<t_{2},...<t_{N}=T-t$, with $N=250$ for a portfolio with time interval $T-t= 1 \ \mbox{year}$. For instance, if we consider a simple portfolio composed of short positions on $250$ calls with a maturity at each $t_{k}$ and a strike $K$. The loss at time $t_{N}=1 \mbox{ year}$ can be written: 
$$
L=\sum_{k=1}^{N} e^{r(t_{N}-t_{k})}(S_{t_{k}}-K)_{+}-e^{r t_{N}}C_{0}^{k},
$$

\noindent where $C_{0}^{i}$ denotes the price of the call of maturity $t_{i}$ and strike $K$, with 
$$
S_{t_{k+1}}=S_{t_{k}}e^{(r-\frac{\sigma^{2}}{2})(t_{k+1}-t_{k})+\sigma \sqrt{(t_{k+1}-t_{k})}Z_{k}}.
$$

\noindent So that, $X=Z=(Z_{1},...,Z_{250})$ is a Gaussian vector with $d=250$. Consequently, with our above procedure, $\theta_{n}$ and $\mu_{n}$ are two random vectors of dimension $d$ and we have to control the growth of each component. If one grows too fast and take too high values, it may provides bad performance and bad estimates of both VaR and CVaR. To circumvent this problem, one can reduce the dimension of the problem by choosing the same shift parameters for several dates, \textit{$i.e.$} for instance 
$$
\theta_{n}=(\underbrace{\theta^{1}_{n},..,\theta^{1}_{n}}_{\mbox{10 times}},...,\underbrace{\theta^{25}_{n},..,\theta^{25}_{n}}_{\mbox{10 times}}).
$$

\noindent Now, we can run the IS algorithm for $\theta^{1},...,\theta^{25}$ so that, we have to deal with a procedure in dimension $25$. It is sub-optimal with respect to the procedure in dimension $250$ but it is more tractable. Another relevant example is a portfolio composed by only one barrier option, for instance a Down \& In Call option 
$$
\varphi(X)=(X_{T}-K)_{+} \mbox{\bf 1 }_{\left\{\min_{\left\{0\leq t \leq T\right\}}X_{t} \leq L\right\}}
$$

\noindent where the underlying $X$ is a process solution of the path-dependent SDE \begin{equation}
\text{d}X_{t}=b(X_{t})\ \mbox{d}t+\sigma(X_{t})\ \mbox{d}W_{t}, \mbox{  } X_{0}=x \in \mathbb{R}^{d},
\label{SDE}
\end{equation}

\noindent $W=(W_{t})_{t\in [0,T]}$ being a standard Brownian motion. A naive approach is to discretize~\eqref{SDE} by an Euler-Maruyama scheme $\bar{X}=(\bar{X}_{t_{k}})_{k\in \left\{0,...,n\right\}}$ 
$$
\bar{X}_{t_{k+1}}=\bar{X}_{t_{k}}+b(\bar{X}_{t_{k}})(t_{k+1}-t_{k})+\sigma(\bar{X}_{t_{k}})(W_{t_{k+1}}-W_{t_{k}}), \hspace*{0.2cm} \bar{X}_{0}=x_{0}\in \mathbb{R}.
$$

\noindent This kind of approximation is known to be poor for this kind of options. In this case, our IS parameters $\theta$ and $\mu$ are $n$-dimensional vectors which correspond to the number of steps in the Euler scheme. Now, if you consider a portfolio composed by several barrier options with different underlyings, the dimension can increase greatly and becomes an important issue, so that our first IS procedure is no longer acceptable and tractable. To overcome this problem, the idea is to shift the entire distribution of $X$ in~\eqref{SDE} thanks to a Girsanov transformation. This last case is analyzed and investigated in \cite{lemairepages}. It can be adapted to our framework (see \cite{frikha} for further developments).

\end{section}


\begin{section}{Numerical examples}
For the sake of simplicity, we focus in this section on the finite dimensional setting and on the computation of the regular CVaR$_{\alpha}$ ($\Psi \equiv Id$). We first consider the usual Gaussian framework in which the exponential change of measure coincide with the mean translation change of measure. Then we illustrate the algorithm~\eqref{EsVaRCVaR} in a simple case.

\subsection{Gaussian framework} In this setting, $X\sim\mathcal{N}(0,I_{d})$ and $p$ is given by 
$$
p(x)=(2 \pi)^{-\frac{d}{2}} e^{-\frac{|x|^{2}}{2}}, \hspace{0.4cm} x\in \mathbb{R}^{d},
$$

\noindent so that (B3) and (B4) are satisfied with $\rho=\frac{1}{2}$ and $b=2$. In this setting, we already noticed that 
\begin{eqnarray*}
L_{3}\left(\xi,\theta,x\right) & := & \mbox{\bf 1}_{\left\{\varphi(X-\theta)
\geq \xi\right\}} (2\theta-x), \\
L_{4}(\xi,\mu,x) & := & \frac{1}{1+G(-\mu)+\xi^2}
\left(\varphi(X-\mu)-\xi\right)_{+}^{2}(2\mu-x).
\end{eqnarray*}

\noindent Moreover, we use a stepwise constant sequence $\alpha_{n}$ that slowly converges toward $\alpha$ as proposed in Section 3.3. We consider three different portfolios of options (puts and calls) on 1 and 5 underlying assets (except for the last case). In the third case, we study the behaviour of a portfolio composed by a power plant that produces electricity from gas with short positions in calls on electricity. The assets are modeled as geometric Brownian motions for the first
two examples. In the third example, the assets (electricity and gas
day-ahead prices) are modeled as exponentials of an Ornstein-Uhlenbeck process. This last derivative is priced using an approximation of Margrabe formulae (see e.g.~\cite{margrabe}). We assume an annual risk free interest rate of 5\%. In each example, we use three different values of the confidence level $\alpha= 95\%, \ 99\%, \ 99.5\%$, which are specified in the Tables. We use the following test portfolios:

\begin{example} Short position in one put with strike $K=110$ and maturity $T=1$ year on a stock with initial value $S_{0}=100$ and volatility $\sigma=20\%$. The loss is given by
$$
\varphi_{1}(X):= (K-S_{T})_{+}-e^{rT}P_{0}
$$ 

\noindent with 
$$
S_{T}:=S_{0}e^{\left(\left(r-\frac{\sigma^{2}}{2}\right)T+\sigma \sqrt{T}X\right)}
$$

\noindent where $X\sim\mathcal{N}(0,1)$ and $P_{0}$ is the initial price at which the put option was sold (it is approximately equal to 10.7). The dimension $d$ of the structural vector $X$ is equal to 1. The numerical results are reported in Table 1. 

\end{example}

\begin{example} Short positions in 10 calls and 10 puts on each of the five underlying assets, all options having the same maturity 0.25 year. The strikes are set to 130 for calls, to 110 for puts and the initial spot prices to 120. The underlying assets have a volatility of 20\% and are assumed to be uncorrelated. The dimension $d$ of the structural vector $X$ is equal to 5. The numerical results are reported in Table 2. 

\end{example}

\begin{example} Short position in a power plant that produces electricity day by day with a maturity of $T=1$ month and 30 long positions in calls on electricity day-ahead price with the same strike $K=60$. Electricity and gas initial spot prices are $S^{e}_{0}=40 \ \$ /\textnormal{MWh}$ and $S^{g}_{0}=3 \ \$ /\textnormal{MMBTU}$ (\textnormal{BTU}: British Thermal Unit) with a Heat Rate equals $h_{R}=10 \ \textnormal{BTU}/\textnormal{kWh}$ and generation costs $C=5 \ \$/\textnormal{MWh}$. The two spot prices have a correlation of 0.4. The payoff can be written
$$
\varphi_{3}(X)= \sum_{k=1}^{30} \left( e^{r(T-t_{k})}\left(S^{e}_{t_{k}}-h_{R} S^{g}_{t_{k}}-C\right)_{+}-P^{c}_{0}e^{rT}\right)+\left(e^{rT}C_{0}-e^{r(T-t_{k})}\left(S^{e}_{t_{k}}-K\right)_{+}\right)
$$

\noindent where $P^{c}_{0}$ is a proxy of the price of the option on the power plant and is equal to 149.9 and $C_{0}$ is the price of the call options which is equal to 3.8. This is a sum of spark spread options where we decide to exchange gas and electricity each day during one month. The dimension $d$ of the structural vector $X$ is equal to 60. The numerical results are reported in Table 3. 
\end{example}

 The results displayed in the following tables correspond to the VaR, the CVaR and the variance reduction ratios estimations for both VaR and CVaR procedure using a number of steps specified in the first column, still for the same three levels of $\alpha$. The variance ratios correspond to the ratio of an estimation of the asymptotic variance using the averaging procedure of~\eqref{VaRCVaRequalsteps} divided by an estimation of the asymptotic variance using the averaging procedure of~\eqref{algoglob}: VR$_{\text{VaR}}$ corresponds to the variance reduction ratio of the VaR estimate and VR$_{\text{CVaR}}$ corresponds to the variance reduction ratio of the CVaR estimate. The results emphasize that the IS procedure yields a very significant, sometimes huge variance reduction especially when $\alpha$ is closed to 1. 

\noindent In the three examples, we define the step sequence by $\gamma_{n}=\frac{1}{n^{\beta}+100}$ where $\beta=\frac{3}{4}$. 

\begin{table}[ht]
\caption{Example 1 Results}
\centering
\begin{tabular}{c c c c c c}
\hline\hline Number of steps & $\alpha$ & VaR & CVaR &
VR$_{\text{VaR}}$ & VR$_{\text{CVaR}}$\\ [0.5ex]

\hline
10 000 & 95\%   & 24.6 & 29.9 & 5.5  & 30.5 \\
  & 99\%        & 34.4 & 37.5 & 11.1 & 125.3 \\
  & 99.5\%      & 37.8 & 41.4 & 13.4 & 192.9 \\
100 000 & 95\%  & 24.6 & 30.4 & 6.6  & 32.2 \\
    & 99\%      & 34.2 & 37.9 & 11.5 & 127.9 \\
  & 99.5\%      & 37.3 & 40.7 & 15.1 & 185 \\
500 000 & 95\%  & 24.6 & 30.3 & 7.7  & 31.3 \\
    & 99\%      & 34.2 & 38   & 14.6 & 118.4 \\
  & 99.5\%      & 37.3 & 40.5 & 15.5 & 184 \\ [1ex]
\hline
\end{tabular}
\label{Portfolio 1}
\end{table}

\begin{table}[ht]
\caption{Example 2 Results}
\centering
\begin{tabular}{c c c c c c}
\hline\hline Number of steps & $\alpha$ & VaR & CVaR &
VR$_{\text{VaR}}$ & VR$_{\text{CVaR}}$\\ [0.5ex]

\hline
10 000 & 95\%  & 339   & 440.5 & 6.5  & 14.9 \\
  & 99\%       & 493.1 & 561.4 & 10.1 & 24.3 \\
  & 99.5\%     & 540.1 & 606.4 & 18.2 & 37.9\\
100 000 & 95\% & 349.8 & 439.7 & 6.7  & 17 \\
    & 99\%     & 495.7 & 563.8 & 11.3 & 28.6 \\
  & 99.5\%     & 544.8 & 607.8 & 18.9 & 40.3 \\
500 000 & 95\% & 352.4 & 439.6 & 6.8  & 17.3 \\
    & 99\%     & 495.2 & 563   & 11.1 & 27.7 \\
  & 99.5\%     & 545.3 & 608.4 & 19.2 & 37 \\ [1ex]
\hline
\end{tabular}
\label{Portfolio 2}
\end{table}

\begin{table}[ht]
\caption{Example 3 Results}
\centering
\begin{tabular}{c c c c c c}
\hline\hline Number of steps & $\alpha$ & VaR & CVaR &
VR$_{\text{VaR}}$ & VR$_{\text{CVaR}}$\\ [0.5ex]

\hline
10 000 & 95\%   & 115.7 & 150.5 & 3.4  & 6.8 \\
  & 99\%        & 169.4 & 196   & 8.4  & 12.9 \\
  & 99.5\%      & 186.3 & 213.2 & 13.5 & 20.3\\
100 000 & 95\%  & 118.7 & 150.5 & 4.5  & 8.7 \\
    & 99\%      & 169.4 & 195.4 & 12.6 & 17.5 \\
  & 99.5\%      & 188.8 & 212.9 & 15.6 & 29.5 \\
500 000 & 95\%  & 119.2 & 150.4 & 5    & 9.2 \\
    & 99\%      & 169.8 & 195.7 & 13.1 & 18.6 \\
  & 99.5\%      & 188.7 & 212.8 & 17   & 29  \\ [1ex]
\hline
\end{tabular}
\label{Portfolio 3}
\end{table} 

\subsection{Esscher transform: the NIG distribution}
Now, we consider a simple case of portfolio composed by a long position on a Call option with strike $K=0.6$ and maturity $T=1$ year, where the underlying is $e^{X_{T}}$ ($X_{0}=0$), where $X_{T}$ is a Normal Inverse Gaussian (NIG) variable, $X_{T}\sim \textnormal{NIG}(\alpha,\beta,\delta,\mu)$, $\alpha>0$, $|\beta|\leq \alpha$, $\delta>0$, $\mu \in \mathbb{R}$. Its density is given by
$$
p_{X_{T}}(x,\mbox{ }\alpha,\mbox{ }\beta,\mbox{ }\delta,\mbox{ }\mu):=\frac{\alpha \delta K_{1}(\alpha \sqrt{\delta^{2}+(x-\mu)^{2}})}{\pi\sqrt{\delta^{2}+(x-\mu)^{2}}}e^{\delta \gamma +\beta(x-\mu)},
$$

\noindent where $K_{1}$ is a modified Bessel function of the second kind and $\gamma=\sqrt{\alpha^{2}-\beta^{2}}.$ Note that the generating function of the NIG distribution is given by 
$$
\psi(\theta)=\mu \theta+\delta(\gamma-\sqrt{\alpha^{2}-(\beta+\theta)^{2}}),
$$
\noindent and is not well defined for every $\theta \in \mathbb{R}$, so that we change the algorithm parametrization (see section 4.3 of \cite{lemairepages}). The loss of the portfolio can be written $L=\varphi_{4}(X_{T})=50(e^{X_{T}}-K)_{+}-e^{rT}C_{0}$. Note that the price $C_{0}$ is computed by a crude Monte Carlo  and is approximately equal to 42. The parameters of the NIG random variable $X_{T}$ are $\alpha=2.0,\mbox{ }\beta=0.2,\mbox{ }\delta=0.8,\mbox{  }\mu=0.04$. We want to compare the variance reduction achieved by the translation of the mean (see section 3.1) and the one achieved by the Esscher Transform (see section 4.1). In the Robbins-Monro procedure, we define the step sequence by $\gamma_{n}=\frac{1}{n^{\beta}+100}$ where $\beta=\frac{3}{4}$.

\medskip
\noindent $\rhd$ \textsl{Translation case.} The functions $L_{3}$ and $L_{4}$ of the IS procedure are defined by: 
\begin{eqnarray*}
L_{3}(\xi,\theta,X) & := & e^{-2|\theta|}\mbox{\bf 1}_{\varphi(X-\theta)}\frac{p{'}(X-2\theta)}{p(X)}\left(\frac{p(X-\theta)}{p(X-2\theta)}\right)^{2}, \\
L_{4}(\xi,\mu,X) & := & \frac{e^{-2|\mu|}}{1+G(-\mu)+\xi^{2}}(\varphi(X-\mu)-\xi)^{2}_{+}\frac{p{'}(X-2\mu)}{p(X)}\left(\frac{p(X-\mu)}{p(X-2\mu)}\right)^{2},
\end{eqnarray*}

\noindent where $p'$ is easily obtained using the relation on the modified Bessel function $K_{1}'(x)=\frac{1}{x}K_{1}(x)-K_{2}(x).$

\medskip
\noindent $\rhd$ \textsl{Esscher Transform.} In this approach, the functions $L_{3}$ and $L_{4}$ are defined by 
\begin{eqnarray*}
L_{3}(\xi,\theta,X) & := & \mbox{\bf 1}_{\varphi(X^{(-\theta)})\geq \xi}(\nabla \psi(\theta)-X^{(-\theta)}), \vspace*{.15cm} \\
L_{4}(\xi,\mu,X) & := & \frac{e^{-|\mu|}}{1+\xi^{2}}(\varphi(X^{(-\mu)})-\xi)^{2}_{+}(\nabla \psi(\mu)-X^{(-\mu)}),
\end{eqnarray*}

\noindent where $X^{(\pm \theta)}\sim \textnormal{NIG}(\alpha,\beta\pm \theta, \delta,\mu).$

\noindent Table 4 compares the variance reduction ratios of the VaR$_{\alpha}$ and CVaR$_{\alpha}$ algorithms achieved by the translation of the mean (VR$^{tr}_{VaR}$ and VR$^{tr}_{CVaR}$) and the one achieved by the Esscher Transform (VR$^{es}_{VaR}$ and VR$^{es}_{CVaR}$).

\begin{table}[ht]
\caption{Example 4 Results}
\centering
\begin{tabular}{c c c c c c c c}
\hline\hline Number of steps & $\alpha$ & VaR & CVaR &
VR$^{tr}_{\text{VaR}}$ & VR$^{tr}_{\text{CVaR}}$ & VR$^{es}_{\text{VaR}}$ & VR$^{es}_{\text{CVaR}}$\\ [0.5ex]
\hline
10 000 & 95\%   & 85.8  & 215.7 & 5   & 10   & 4.2  & 58.8\\
  & 99\%        & 217   & 518   & 6   & 12   & 8    & 60 \\
  & 99.5\%      & 304   & 748   & 8   & 25   & 8.9  & 110 \\
100 000 & 95\%  & 87.2  & 215.1 & 5   & 12   & 4.5  & 60\\
    & 99\%      & 218   & 521   & 5   & 12   & 8.2  & 70\\
  & 99.5\%      & 303.5 & 747.8 & 7   & 30   & 12   & 100\\
500 000 & 95\%  & 87.9  & 215.6 & 5   & 9    & 5    & 57 \\
    & 99\%      & 227   & 518.9 & 5.5 & 11.8 & 11.5 & 68 \\
  & 99.5\%      & 312.8 & 741.8 & 6   & 31   & 10   & 123 \\ [1ex]
\hline
\end{tabular}
\label{Portfolio 4}
\end{table} 

The IS procedure is very efficient when $\mathbb{P}(\varphi(X) \geq \xi^{*}_{\alpha})=1-\alpha$ is close to zero and becomes more and more efficient as $\alpha$ grows to 1. Even for the complex portfolio considered in Example 3, where $X$ is a Gaussian vector with $d=60$, it is possible to achieve a great variance reduction for both VaR$_{\alpha}$ and CVaR$_{\alpha}$.

 We observed that   IS based on Esscher transform is well adapted to distributions with heavy tails ($i.e.$  heavier tails than the normal distribution). It is therefore suitable when large values are more frequent than for the normal distribution, as it is the case when the vector $X$ is a NIG random variable. Indeed, in this setting, the IS parameters modify the parameter $\beta$ which controls the asymmetric shape of the NIG distribution. We think that the IS procedure by Esscher transform outperforms the IS procedure by mean translation when the IS parameter impacts on the symmetry of the distribution. 

\end{section}

\begin{section}{Concluding remarks.} In this article, we propose a recursive procedure to compute efficiently the
Value-at-Risk and the Conditional Value-at-Risk using the same innovation for both procedures. In our approach, for a given risk level $\alpha$, the VaR$_{\alpha}$ and the CVaR$_{\alpha}$ are estimated simultaneously by a regular RM algorithm. Ruppert and Polyak's averaging principle provides an asymptotically efficient procedure. The estimates satisfy a Gaussian CLT. However, due to the slow convergence of the global procedure since we are interested in rare events, the regular version of this algorithm cannot be used in practice. To speed-up and thus greatly reduce the number of scenarios, we devise an unconstrained adaptive IS procedure. The resulting procedure provides estimates that satisfy a CLT with minimal variances. To optimize the move to the critical risk area, the risk level $\alpha$ can be temporarily replaced by a slowly increasing level $\alpha_{n}$ (stepwise constant in practice) converging to $\alpha$. This produces a VaR \emph{companion procedure} $(\hat{\xi}_{n})_{n\geq1}$ that controls the IS change of measure parameters $(\hat{\theta}_{n},\hat{\mu}_{n})$. Numerically speaking, the resulting procedure converges efficiently and can drastically reduce the variance. It is possible to extend the methods to portfolio whose losses depend on a general diffusion process, using Girsanov transform to introduce a potentially infinite dimensional variance reducer. Finally, we aim at extending the method by implementing  low-discrepancy sequences  in our procedure instead of pseudo-random numbers. Preliminary numerical experiments showed a significant improvement of the convergence rate . This also raises interesting theoretical problems. See~\cite{LaPaSa} for some first theoretical results in that direction in a  one-dimensional framework and~\cite{frikha} for further developments in higher dimensional setting.
\end{section}

\end{document}